\documentclass{article}
\RequirePackage{amsmath,amsthm,amssymb,stmaryrd,cite,enumitem}
\usepackage[activate={true,nocompatibility},final,tracking=true,kerning=true,spacing=true,factor=1100,stretch=10,shrink=10,spacing=nonfrench]{microtype}
\RequirePackage{tikz,pgf}
\usetikzlibrary{decorations,decorations.markings}
\theoremstyle{plain}
\newtheorem{Proposition}{Proposition}
\newtheorem{Lemma}{Lemma}
\newtheorem{Definition}{Definition}
\newtheorem{Theorem}{Theorem}
\DeclareMathOperator{\vect}{vec}
\DeclareMathOperator{\diag}{diag}
\allowdisplaybreaks
\newcommand{\dumbbell}{{\textcolor{black}{\bullet\!\!-\!\!\bullet}}}
\newcommand{\ddumbbell}{{\ensuremath{
	\substack{\dumbbell\\[-9.6pt]\dumbbell}
}}}
\newcommand{\arro}{{\ensuremath{
	\frac{\ }{\ }\hspace{-2pt}\triangleleft}
}}
\newcommand{\pfour}{{\ensuremath{
	\vee\!\vee
}}}
\tikzstyle{dir}= [postaction={decorate,
	decoration={markings,
	mark=at position .65
	with {\arrow[scale=.8]{stealth}}}}]
\tikzstyle{dirs}= [postaction={decorate,
	decoration={markings,
	mark=at position .65
	with {\arrow[scale=.7]{stealth}}}}]
\tikzstyle{nd} = [circle, fill=black,
	inner sep=0pt,
	minimum width=3 pt]
\tikzstyle{bnd} = [circle,fill=black,
	inner sep=.4pt,
	minimum width=3 pt,
	text=white]
\tikzstyle{rnd} = [circle, fill=black!30,
	inner sep=0pt,
	minimum width=3 pt]
\tikzstyle{brnd} = [circle,fill=black!30,
	inner sep=1pt,
	minimum width=6 pt,
	text=black]
\newcommand{\TexttriangleS}{%
\raisebox{-0.45mm}{\!\!
\tikz[scale=.2,clip]{\draw[thick]
(210:.6) node[nd] {}--
(  90:.6) node[nd]{}--
( -30:.6) node[nd]{}--
(210:.6) node[nd]{};
}}
}
\newcommand{\Texttriangle}{%
\raisebox{-0.45mm}{\!\!
\tikz[scale=.2,clip]{\draw[thick]
(210:.6) node[nd] {}--
(  90:.6) node[nd]{}--
( -30:.6) node[nd]{}--
(210:.6) node[nd]{};
}}\hspace{-.9mm}
}
\newcommand{\TexttriangleRS}{%
\raisebox{-0.45mm}{\!\!
\tikz[scale=.2,clip]{\draw[thick]
(210:.6) node[rnd] {}--
(  90:.6) node[nd]{}--
( -30:.6) node[nd]{}--
(210:.6) node[rnd]{};
}}
}
\newcommand{\TexttriangleR}{%
\raisebox{-0.45mm}{\!\!
\tikz[scale=.2,clip]{\draw[thick]
(210:.6) node[rnd] {}--
(  90:.6) node[nd]{}--
( -30:.6) node[nd]{}--
(210:.6) node[rnd]{};
}}\hspace{-.9mm}
}
\newcommand{\TexttriangleRR}{%
\raisebox{-0.45mm}{\!\!
\tikz[scale=.2,clip]{\draw[thick]
(210:.6) node[rnd] {}--
(  90:.6) node[rnd]{}--
( -30:.6) node[nd]{}--
(210:.6) node[rnd]{};
}}
}
\newcommand{\Textsquare}{%
\raisebox{-0.45mm}{\!\!
\tikz[scale=.2,clip]{\draw[thick]
(        45:.65) node[nd]{}--
(  45+90:.65) node[nd]{}--
(45+180:.65) node[nd]{}--
(45+270:.65) node[nd]{}--
(        45:.65) node[nd]{};
}}\hspace{-.9mm}
}

\newcommand{\TextpathtwoR}{%
\raisebox{-0.45mm}{\!\!
\tikz[scale=.2,clip]{\draw[thick]
( -30:.6) node[nd]{}--
(210:.6) node[rnd]{}--
(  90:.6) node[nd]{};
}}\hspace{-.9mm}
}
\newcommand{\EqtriangleR}{%
\raisebox{-3mm}{
\tikz[scale=.5,clip]{\draw[very thick]
(210:.6) node[bnd]{}--
(  90:.6) node[bnd]{{\scriptsize 2}}--
( -30:.6) node[bnd]{{\scriptsize 3}}--
(210:.6) node[brnd]{{\scriptsize 1}};
}}\hspace{1.5mm}
}
\newcommand{\FigGraph}{%
\begin{tikzpicture}[thick,scale=.54]
	\tikzstyle{pop}=[circle, draw, color=black!50,fill=black!30,text = black,
					inner sep=.0pt, minimum width=7pt]
	\tikzstyle{popi}=[circle, draw, color = black, fill=black,
					inner sep=0pt, minimum width=4pt]
	\tikzset{pil/.style={thick,color=black,shorten >=2.25pt}}
	\tikzset{pul/.style={thick,color=black!30,shorten >=2.25pt}}
	\draw[fill=black!10,draw=white]
		(210+36:3)--(330-36:3)--(330:2)--(0:0)--(210:2)--cycle;
	\draw[fill=black!10,draw=white]
		(0:0)--(90:2)--(90+36:3)--(45:3)--(90:2)--(6:3)--cycle;
	\draw {(330-36:3) node[popi] {}edge[pul] (210:2) node[popi] {}};
	\draw {(210+36:3) node[popi] {} edge[pul] (330-36:3) node[popi] {}};
	\draw {(0:0) node[popi] {} edge[pul] (330-36:3) node[popi] {}};
	\draw {(330:2) node[popi] {} edge[pul] (330-36:3) node[popi] {}};
	\draw {(0:0) node[popi] {} edge[pul] (210:2) node[popi] {}};
	\draw {(210:2) node[popi] {} edge[pul] (330:2) node[popi] {}};
	\draw {(210:2) node[popi] {} edge[pul] (36+210:3) node[popi] {}};
	\draw {(0:0) node[popi] {} edge[pul] (330:2) node[popi] {}};
	\draw {(330-36:3) node[pop] {{\scriptsize{w}}}};
	\draw {(0:0) node[popi] {} edge[pul] (90:2) node[popi] {}};
	\draw {(90:2) node[popi] {} edge[pul] (90+36:3) node[popi] {}};
	\draw {(90:2) node[popi] {} edge[pul] (6:3) node[popi] {}};
	\draw {(90+36:3) node[popi] {} edge[pul] (45:3) node[popi] {}};
	\draw {(45:3) node[popi] {} edge[pul] (90:2) node[popi] {}};
	\draw {(6:3) node[popi] {} edge[pul] (0:0) node[popi] {}};		
	\draw{(90:2) node[pop] {{\scriptsize{u}}}};
	\draw {(0:0) node[pop] {{\scriptsize{v}}}};
	\draw {(210+36:3) node[popi] {}};
	\draw {(330:2) node[popi] {} edge[pil] (36+330:3) node[popi] {}};
	\draw {(210:2) node[popi] {} edge[pil] (210-36:3) node[popi] {}};
	\draw {(330:2) node[popi] {} edge[pil] (45:3) node[popi] {}};
	\draw {(210-36:3) node[popi] {} edge[pil] (90+36:3) node[popi] {}};
	\draw {(210:3.75) node[popi] {} edge[pil] (210:2) node[popi] {}};
\end{tikzpicture}
}
\newcommand{\EqPathThreeRR}{%
\raisebox{-1pt}{
\tikz[scale=1,clip]{\draw[very thick,black!30,dir]
(180:.6) node[bnd]{{\scriptsize 1}}--
(    0, 0) node[bnd]{{\scriptsize 2}};
\draw[very thick,black!30,dir]
(    0, 0) node[bnd]{{\scriptsize 2}}--
(    0:.6) node[bnd]{{\scriptsize 3}};
}}\hspace{.75mm}
}

\newcommand{\EqTriangleRRd}{%
\raisebox{-3mm}{
\tikz[scale=.5,clip]{\draw[very thick,black!30,dir]
(210:.6) node[bnd]{}--
(  90:.6) node[bnd]{{\scriptsize 2}};
\draw[very thick]
(  90:.6) node[bnd]{{\scriptsize 2}}--
( -30:.6) node[bnd]{{\scriptsize 3}};
\draw[very thick,black!30,dir]
( -30:.6) node[bnd]{{\scriptsize 3}}--
(210:.6) node[bnd]{{\scriptsize 1}};
}}\hspace{1.5mm}
}
\newcommand{\EqTriangleRR}{%
\raisebox{-3mm}{
\tikz[scale=.5,clip]{\draw[very thick]
(210:.6) node[bnd]{}--
(  90:.6) node[brnd]{{\scriptsize 2}}--
( -30:.6) node[bnd]{{\scriptsize 3}}--
(210:.6) node[brnd]{{\scriptsize 1}};
}}\hspace{1.5mm}
}
\newcommand{\EqPathFourRRd}{%
\raisebox{-1pt}{
\tikz[scale=.7,clip]{\draw[very thick,black!30,dir]
(180:  .6) node[bnd]{{\scriptsize 1}}--
(    0,   0) node[bnd]{{\scriptsize 2}};
\draw[very thick]
(    0,   0) node[bnd]{{\scriptsize 2}}--
(    0:  .6) node[bnd]{{\scriptsize 3}};
\draw[very thick,black!30,dir]
(    0:  .6) node[bnd]{{\scriptsize 3}}--
(    0:1.2) node[bnd]{{\scriptsize 4}};
}}\hspace{.75mm}
}

\newcommand{\EqSquareRRd}{%
\raisebox{-2.5mm}{
\tikz[scale=.45,clip]{\draw[very thick,black!30,dirs]
(225:.65) node[bnd]{}--
(135:.65) node[bnd]{{\scriptsize 2}};
\draw[very thick]
(135:.65) node[bnd]{{\scriptsize 2}}--
(  45:.65) node[bnd]{{\scriptsize 3}};
\draw[very thick,black!30,dirs]
(  45:.65) node[bnd]{{\scriptsize 3}}--
( -45:.65) node[bnd]{{\scriptsize 4}};
\draw[very thick]
( -45:.65) node[bnd]{{\scriptsize 4}}--
(225:.65) node[bnd]{{\scriptsize 1}};
}}\hspace{1.5mm}
}

\newcommand{\EqCompleteRR}{%
\raisebox{-3mm}{
\tikz[scale=.5,clip]{\draw[very thick,black!30,dirs]
(225:.65) node[bnd]{}--
(135:.65) node[bnd]{};
\draw[very thick,black!30,dirs]
(  45:.65) node[bnd]{}--
( -45:.65) node[bnd]{};
\draw[very thick]
(135:.65) node[bnd]{}--
( -45:.65) node[bnd]{};
\draw[very thick]
(  45:.65) node[bnd]{}--
(225:.65) node[bnd]{};
\draw[very thick]
(135:.65) node[bnd]{{\scriptsize 2}}--
(  45:.65) node[bnd]{{\scriptsize 3}};
\draw[very thick]
( -45:.65) node[bnd]{{\scriptsize 4}}--
(225:.65) node[bnd]{{\scriptsize 1}};
}}\hspace{1.5mm}
}

\title{Fast counting of medium-sized rooted subgraphs}
\author{P-A. Maugis, S. C. Olhede and P. J. Wolfe\\
{\em University College London}\footnote{This work was supported in part by the US Army Research Office under Multidisciplinary University Research Initiative Award 58153-MA-MUR; by the US Office of Naval Research under Award N00014-14-1-0819; by the UK Engineering and Physical Sciences Research Council under Mathematical Sciences Leadership Fellowship EP/I005250/1, Established Career Fellowship EP/K005413/1, Developing Leaders Award EP/L001519/1, and Award EP/N007336/1; by the UK Royal Society under a Wolfson Research Merit Award; and by Marie Curie FP7 Integration Grant PCIG12-GA-2012-334622 and the European Research Council under Grant CoG 2015-682172NETS, both within the Seventh European Union Framework Program.}}\date{}
\begin{document}

\maketitle
\begin{abstract}
We prove that counting copies of any graph $F$ in another graph $G$ can be achieved using basic matrix operations on the adjacency matrix of $G$. Moreover, the resulting algorithm is competitive for medium-sized $F$: our algorithm recovers the best known complexity for rooted 6-clique counting and improves on the best known for 9-cycle counting. Underpinning our proofs is the new result that, for a general class of graph operators, matrix operations are homomorphisms for operations on rooted graphs.
\end{abstract}

\section{Introduction}
Counting the subgraphs that compose a graph---the number of triangles (\TexttriangleS\nobreak\hspace{.05em plus .01em}), squares (\Textsquare), trees, and so on---is a key primitive of network data analysis. For this reason, many algorithms have been proposed to tackle the subgraph counting problem (a problem also known as graphlet, motif or shape counting or census); see~\cite{chiba1985arboricity,eisenbrand2004complexity,latapy2008main,vassilevska2009clique,gonen2009,marcus2010efficient,kowaluk2013counting,kolda2014counting,bjorklund2014counting,hocevar2014,floderus2015detecting,williams2015finding,jha2015path,talukder2016distributed,ortmann2016quad}, to cite but a few. Driving these works is the realization that larger subgraphs are harder to count, but crucial to the analysis of network data~\cite{solava2012graphlet,Ali2014alignment}.

Here we present the surprising result that despite the complexity of the subgraph counting problem, rooted subgraph counting---i.e., tabulating not only the number of subgraphs but also where they are in the graph---may be efficiently achieved using basic matrix operations. Specifically, we first show that basic matrix operations are primitives of the subgraph counting problem:
\begin{Theorem}\label{thm1}\vspace{-.25\baselineskip}
Counting copies of any subgraph around any location in a graph may be achieved using basic matrix operations on the graph's adjacency matrix.\vspace{-.33\baselineskip}
\end{Theorem}
Then, we show that this reduction of subgraph counting to matrix operations yields an efficient counting algorithm for medium-sized subgraphs:
\begin{Theorem}\label{thm2}\vspace{-.25\baselineskip}
Let $G$ be a simple graph with $n$ vertices, $m$ edges and $l$ paths of length 2. Let $\omega_k(p,q)$ be the complexity of computing the product of $k$ square matrices each of size $p$ with $q$ entries. Then, the complexity of counting all edge-rooted subgraphs of order $k$ for $k\leq 6$ in $G$ is $O\big(\omega_4(2m,2l)\big)$, and the complexity of counting all $k$-cycles for $k\leq 9$ in $G$ is $O\big(\omega_6(2m,2l)\big)$.\vspace{-.33\baselineskip}
\end{Theorem}
Using naive sparse matrix product~\cite{yuster2005fast}, we obtain that if $m,l=O(n)$---as in real-world graphs~\cite{Barabasi99,clauset2009power}---then the complexity of medium-sized rooted subgraph counting is $O(n^2)$. More generally, the complexity is $O(\min\{lm,m^\omega\})$, with $\omega$ the complexity exponent of the fast matrix product.

In effect, the larger the counted subgraph is, the more matrix based counting becomes competitive. Thus, although graph specific methods outperform matrix based algorithms for small subgraphs, it appears that for medium-sized subgraphs, matrix based algorithms may be more efficient. Specifically, already for order 6 subgraphs, our algorithm recovers the best known complexity for rooted 6-clique counting~\cite{eisenbrand2004complexity}, while no algorithm matches ours for 9-cycle counting~\cite{alon1997cycles}.

\section{Rooted graph counting}\label{Roots}

We define bi-rooted graphs and connection matrices. The algebraic properties of these objects will be our main tools to prove Theorem~\ref{thm1} and~\ref{thm2} in the next section. Bi-rooted graphs generalize rooted graphs, graphlet orbits~\cite{Przulj2007} and partially labelled graphs~\cite{lovasz2012large}. We use them to count the copies of a subgraph connecting pairs of sets of vertices. Then, we define a connection matrix as the arrays tabulating how many copies of a bi-rooted graph connect pairs of sets of vertices. We find that connection matrices generalize most used graph operators.

Classically, rooted graphs are used to count copies around a given vertex; say the number of \Texttriangle\,\!\! containing a given vertex. Formally, this is done by counting the number of \TexttriangleR\,\!\! where one vertex (in grey) is pinned-down at one vertex in the graph. A vertex's degree is a rooted subgraph count: it is the number of edges attached to this vertex. The local clustering coefficient is also computed using rooted subgraph counts: it is the ratio of the number of closed (\TexttriangleRS\nobreak\hspace{.01em plus .01em}) over open (\TextpathtwoR) triangles attached to a vertex. More generally, counting rooted subgraphs is recognized as a fundamental tool from both theoretical~\cite{rucinski1986balanced,lovasz2014automorphism} and applied~\cite{isham2011spread,Przulj2007,solava2012graphlet,Ali2014alignment} viewpoints.

To count copies connecting groups of vertices, we define bi-rooted graphs as having two roots. Therefore, we write:
\begin{Definition}[Bi-rooted graph]\label{root}
\vspace{-.25\baselineskip}
A bi-rooted graph is an ordered triple $[F,r,s]$, where $F=(V(F),E(F))$ is a graph and $r,s$ are two tuples of $V(F)$.\vspace{-.33\baselineskip}
\end{Definition}
For example, denoting $K_3 = (\{1,2,3\},\{12,23,31\})$ the triangle, we have
\vspace{-.33\baselineskip}\[
\EqtriangleR = [K_3,1,\emptyset];\ \ \!%
\EqTriangleRR = [K_3,12,\emptyset].
\vspace{-.33\baselineskip}\]
With this notation, for $F$ a graph and $v$ one of its vertices, $[F,v,\emptyset]$ recovers the most common notion of rooted graphs, or graphlet orbits~\cite{Przulj2007}. Partially labelled graphs~\cite{lovasz2012large} are recovered by $[F,r,\emptyset]$ for $r$ any tuple of vertices in $F$.

\begin{figure}
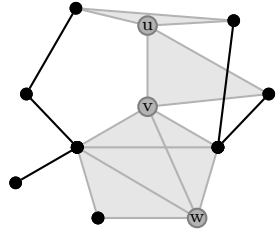

\centering\ 
\begin{minipage}[c]{0.52\textwidth}
        \caption{\label{CountingExample}
        Examples of rooted subgraph counts. The vertex labeled ``u'' has a \protect\TexttriangleR\,\! count of 2 while the vertex labeled ``v'' has a \protect\TexttriangleR\,\! count of 4. The pair of vertices labeled ``v'' and ``w'' has a \protect\TexttriangleRR\,\! count of 2.}
\end{minipage}\hfill
\begin{minipage}[c]{0.45\textwidth}\centering
        \FigGraph
\end{minipage}
\ \vspace{.4\baselineskip}\\
{\centering{
\hfill{\Huge{--------------------------------------------}}\hfill}}
\vspace{-\baselineskip}
\end{figure}

We can now count the number of copies of a subgraph $F$ connecting groups of vertices using bi-rooted graphs. In Fig.~\ref{CountingExample} we give examples of bi-rooted subgraph counting of \TexttriangleR\,\! and \TexttriangleRR\nobreak\hspace{.001em plus .01em}. Given a graph $G$ and two roots $i,j$ (tuples of $V(G)$), it consists in counting the number of ways $[F,r,s]$ can be embedded in $G$ while mapping the roots $r$ and $s$ onto $i$ and $j$ respectively; i.e., the number of isomorphic copies of $[F,r,s]$ in $[G,i,j]$.

Tabulating all such counts in matrix we obtain a {\em connection matrix}:
\begin{Definition}[Connection matrix]\label{count}
\vspace{-.25\baselineskip}
Fix a graph $G$ and a bi-rooted graph $F=[F,r,s]$. We call connection matrix of $F$ in $G$, and denote $\kappa(F,G)$, the matrix indexed by the $|r|$ and $|s|$-tuples of $V(G)$ such that
\[
\kappa(F,G)_{ij}= \#\big\{[F',i,j]\subset [G,i,j]\,:\,[F',i,j]\equiv [F,r,s]\big\},
\]
where $[F',i,j]\subset [G,i,j]$ if $V(F')\subset V(G)$ and $E(F')\subset E(G)$, while $[F',i,j]\equiv [F,r,s]$ if there exists a bijective map $\phi:V(F)\to V(F')$ such that $\phi(r)=\phi(i)$, $\phi(s)=\phi(j)$ and $\phi(u)\phi(v)\in E(F') \Leftrightarrow uv\in E(F)$.\vspace{-.33\baselineskip}
\end{Definition}
With $r=s=\emptyset$, this definition of counting copies recovers that of counting the not necessarily induced copies of $F$ in $G$. Deducing the number of induced copies from the not necessarily induced copies is a classical transformation; see~\cite{marcus2010efficient,jha2015path,kowaluk2013counting} for some examples. Furthermore, Definitions~\ref{root} and~\ref{count} are such that for $v$ and $v'$ two non-automorphic nodes in $F$, then $[F,v,\emptyset]\not\equiv[F,v',\emptyset]$. Therefore, as in counting graphlet orbits, the location of the root in $F$ cannot be ignored~\cite{Przulj2007}.

Crucially, connection matrices generalize most used graph operators. For instance: with $A$ the adjacency matrix of $G$ and $K_2:=(\{1,2\},\{12\})$, we have $A=\kappa([K_2,1,2],G)$; with $I$ the positive part of the incidence matrix of $G$, we have $I = \kappa([K_2,2,12],G)$; with $B$ the non-backtracking or Hashimoto matrix of $G$~\cite{Krzakala13} and $P^{\ddumbbell}_3 =$ $[(\{1,2,3\},$ $\{12,23\}),$ $12, 23]$, we have $B=\kappa(P^{\ddumbbell}_3,G)$.

Although general, our definitions provide sufficient structure for our proofs. As we will see, by using the algebraic properties of connection matrices and bi-rooted graphs, we build formulas using either the adjacency matrix and the non-backtracking matrix to compute any bi-rooted count.

\section{Main Results}\label{Heur}
We now formally state our two theorems and present a summary of their proofs. Before we begin, observe that the subgraph counting problem is equivalent to computing any $\kappa(F,G)$, as the $\kappa(F,G)_{ij}$ are exactly the number of copies of $F$ in $G$ connecting the tuples of vertices $i$ and $j$. Then, to prove Theorem~\ref{thm1}, we express any $\kappa(F,G)$ using classical matrix operation on the adjacency matrix. Formally:
\vspace{-.75\baselineskip}
\paragraph{Theorem~\protect\ref{thm1}}
\emph{Let $\mathfrak{F}$ be the set of all formulas that take a matrix as sole argument and are built using only: operator, entry-wise and Kronecker products, vectorizations, transposes and constants. Then, with $A(G)$ the adjacency matrix of a graph $G$, we have}
\vspace{-.2\baselineskip}\begin{equation}\label{e1}
\forall F,\ \exists f_F\in\mathfrak{F},\ s.t.,\ \forall G,\ \kappa(F,G)=f_F(A(G)).
\end{equation}

Key to the generality and simplicity of our proof is the observation that operations on connection matrices act as homomorphisms over rooted graphs. Let $\digamma$ be the set of rooted graphs that verify~\eqref{e1}, and for each $F\in\digamma$ call $f_F$ the associated formula. Then, for $[F,r,s], [F',r,s]\in\digamma$, we for instance have that (see Lemma~1 in the Appendix):
\begin{equation}\label{trans}
f_{[F,r,s]}^\top = f_{[F,s,r]},\quad
\vect(f_{[F,r,s]}) = f_{[F,rs,\emptyset]};
\end{equation}
and if $r\cup s=V(F)$, denoting ``$\cdot$'' the entry-wise product,
\begin{equation}\label{hada}
f_{[F,r,s]}\cdot f_{[F',r,s]} = f_{[F\cup F',r,s]}.
\end{equation}
Such properties allow us to prove the following proposition in the Appendix:
\vspace{-.05\baselineskip}\begin{Proposition}\label{hp}
Set $F=[F,r,s]$, $F' = [F',r,s]$ and $F''=[F'',s,t]$ in $\digamma$. Then:\vspace{-.3\baselineskip}
\begin{enumerate}\itemsep0em
\item $\forall r',s'\subset s\cup r,\ [F,r',s']\in\digamma$,
\item $r\cup s=V(F)\Rightarrow
[F\cup F',r,s]\in\digamma$,
\item $r\cup s=V(F),\ s\cup t = V(F'')\Rightarrow
[F\cup F'',r,t]\in\digamma$.
\end{enumerate}
\end{Proposition}\vspace{-.3\baselineskip}

In the proof of Theorem~\ref{thm1} we use Proposition~\ref{hp} to iteratively build formulas for larger and larger $F$, ultimately allowing us to show that any rooted graph $F$ is in $\digamma$. Specifically, the proof proceeds by recurrence on the order of $F$, paralleling the proofs of~\cite{chiba1985arboricity,meeks2016tree,nevsetvril2012sparsity,kowaluk2013counting,demaine2014structural}, where the induction variable is instead the independent set or the tree-width.

We initialize with order 2 graphs. First, observe that with $K_2:=(\{1,2\},\{12\})$, we have $A(G)=\kappa([K_2,1,2],G)$ and therefore $[K_2,1,2]\in\digamma$ with $f_{[K_2,1,2]}$ being the identity. Then, by Proposition~\ref{hp}.1, $[K_2,1,2]\in\digamma$ implies that all other order two rooted graphs are also in $\digamma$, and the initialization is complete. 

For the induction step, we assume that all order at most $k$ rooted graphs are in $\digamma$, and fix a rooted graph $F=[F,r,s]$ of order $k+1$. To show that $F$ is in $\digamma$ we select one of its vertices, say $t$, and use Proposition~\ref{hp} to successively:
\vspace{-.3\baselineskip}
\begin{itemize}\itemsep0em
\item[--] Set $F'$ the subgraph of $F$ obtained by removing all the edges attached to $t$, but having the same vertex set as $F$, and show that $[F',V(F),t]\in\digamma$.\vspace{-.2\baselineskip}
\item[--] Set $F''$ the subgraph of $F$ obtained by keeping only all the edges attached to $t$, but having the same vertex set as $F$, and show that $[F'',V(F),t]\in\digamma$.\vspace{-.2\baselineskip}
\item[--] Show that $[F'\cup F'',V(F),t]$ is in $\digamma$, which yields that $F$ is in $\digamma$.
\end{itemize}\vspace{-.3\baselineskip}
We present a complete proof in the Appendix.

Theorem~\ref{thm1} yields formulas to compute the connection matrix of any rooted graph. However, we can control the complexity of evaluating such formulas by choosing the operations involved carefully. We find that using operations on the non-backtracking matrix~\cite{Krzakala13} $B(\cdot) = \kappa(P^{\ddumbbell}_3,\cdot)$, where $P^{\ddumbbell}_3 = [(\{1,2,3\}, \{12,23\}),$ $12, 23]=\EqPathThreeRR$, we have:
\vspace{-.75\baselineskip}
\paragraph{Theorem~\protect\ref{thm2}*}\emph{Call edge-rooted graph a rooted graph having for root a pair of connected vertices and the empty set. Then, with $\mathcal{F}^{\dumbbell}_{\leq6}$ the set of order at most 6 edge-rooted simple graphs, we have}
\[
\forall F\in\mathcal{F}^{\dumbbell}_{\leq6},\ \exists f_F'\in \mathfrak{F},\ s.t.\ \forall G,\ \kappa(F,G) = f_F'(B(G)).
\]
Our proof consist in presenting one such formula for each $F\in\mathcal{F}^{\dumbbell}_{\leq6}$. These can be found in the Appendix. Then, scrutiny of the formulas shows first that none involve Kronecker products or vectorization, and then that they at most involve the operator product of 4 matrices at least as sparse as $B(G)$. In the Appendix we also present the formula to compute the total 9-cycle count using $B(G)$, which involves the 6th power of $B$. Together, these yield the theorem as stated in the introduction.

To give an intuition as to how these formulas were produced, we present as an example how to build the formula for the edge rooted 6-clique $K_6^{\dumbbell}$. In the following, to simplify notation, we drop the dependence in the underlying graph $G$, so that $B = B(G)$ and $\kappa(F) = \kappa(F,G)$. We also write $[F,r,s]^\top = [F,s,r]$. Then, as long as the union of the roots is the vertex set, using~\eqref{trans} and~\eqref{hada} we have $\kappa(F^\top) = \kappa(F)^\top$ and if furthermore $V(F)=V(F')$, $\kappa(F)\!\cdot\!\kappa(F')=\kappa(F\cup F')$.

To build $\kappa(K_6^{\dumbbell})$, our first step is to evaluate $\kappa(C_4^{\ddumbbell})$, where
\[
C_4^{\ddumbbell}
        = [\{1,2,3,4\},\{12,23,34,41\},12,34]
        = [C_4,12,34]
        = \EqSquareRRd\!.
\]
Then, we will iteratively use identities of the type of~\eqref{hada} to build $K_6^{\dumbbell}$ from several instances of $C_4^{\ddumbbell}$. To do so, we first consider $B^2$. Writing
\begin{align*}\begin{array}{rlcc}
C_3^\ddumbbell&
        = [\{1,2,3\},\{12,23,31\},12,31]
         &=& \EqTriangleRRd,\\
P_4^\ddumbbell&
        =[\{1,2,3,4\},\{12,23,34\},12,34]
        &=&\EqPathFourRRd,
\end{array}\end{align*}
direct computation shows that $B^2 = \kappa(C_3^{\ddumbbell})+\kappa(P_4^{\ddumbbell})$. We now recover $\kappa(C_4^\ddumbbell)$ using the entry-wise product with $B^{2\top}$. Indeed, using~\eqref{trans} and~\eqref{hada}, we have that
\begin{align*}
&B^2\!\cdot\!B^{2\top}
=\big(\kappa(C_3^{\ddumbbell})\!+\!\kappa(P_4^{\ddumbbell})\big)\!\cdot\!\big(\kappa(C_3^{\ddumbbell})\!+\!\kappa(P_4^{\ddumbbell})\big)^\top\\
&\stackrel{\text{Eq.~\eqref{trans}}}{=} 
\kappa(C_3^{\ddumbbell})\!\cdot\!\kappa(C_3^{\ddumbbell\top})\!+\!
\kappa(C_3^{\ddumbbell})\!\cdot\!\kappa(P_4^{\ddumbbell\top})\!+\!
\kappa(P_4^{\ddumbbell})\!\cdot\!\kappa(C_3^{\ddumbbell\top})\!+\!
\kappa(P_4^{\ddumbbell})\!\cdot\!\kappa(P_4^{\ddumbbell\top})\\
&\stackrel{\text{Eq.~\eqref{hada}}}{=} 
\kappa(P_4^{\ddumbbell}\cup P_4^{\ddumbbell\top})
= \kappa(C_4^\ddumbbell),
\end{align*}
observing that the first three terms of the second line are equal to 0 because the structure of the root sets do not match. Set $\bar C_4^\ddumbbell = [C_4,12,43]$, and notice that we can compute $\kappa(\bar C_4^\ddumbbell)$ through a simple reordering of the columns of $\kappa(C_4^\ddumbbell)$. Then, we have $\kappa(C_4^\ddumbbell)\cdot \kappa(\bar C_4^\ddumbbell) = \kappa(C_4^\ddumbbell\cup \bar C_4^\ddumbbell) = \kappa(K_4^{\ddumbbell})$, where
\[K_4^{\ddumbbell}
        = (\{1,2,3,4\},\{12,13,14,23,24,34\},12,34)
        = \EqCompleteRR.
\]
To conclude, we use that a 6-clique can be partitioned into three 4-cliques, each pair of which share an edge disconnected from the third. First we observe that for each pair of edges, the number of copies of $H$---where $H^{\ddumbbell}$ is a pair of 4-cliques connected by an edge---verifies (as $G$ is simple, $xy\in E(G)\Rightarrow yx\in E(G)$)
\[
\big(\kappa(K_4^{\ddumbbell}) \kappa(K_4^{\ddumbbell})\big)_{ij,kl} = 2\sum_{x<y}\kappa(K_4^{\ddumbbell})_{ij,xy}\kappa(K_4^{\ddumbbell})_{xy,kl} = 2 \kappa(H)_{ij,kl},
\vspace{-.333\baselineskip}
\]
with $H^{\ddumbbell}\!=\![\{1,2,...,6\},\{pq: 1\leq p<q\leq 4\}\cup\{pq:3\leq p<q\leq 6\},12,56]$. Then, by~\eqref{hada}, we have that
\[
\kappa(H^{\ddumbbell})\cdot\kappa(K_4^{\ddumbbell}) = \kappa(H^{\ddumbbell}\cup K_4^{\ddumbbell})= \kappa(K_6^{\ddumbbell}),
\]
where $K_6^{\ddumbbell} = (\{1,2,...,6\},\{pq:1\leq p<q\leq 6\},12,56)$. To summarize, we obtained that with ${\bf 1}$ the vector of all 1 in the appropriate dimension, \[
\kappa(K_6^{\dumbbell})=
\big((\kappa(K_4^{\ddumbbell}) \kappa(K_4^{\ddumbbell}))\cdot \kappa(K_4^{\ddumbbell})\big) {\bf 1}/4!,
\]
since there are $4!$ partitions of a 6-clique into three 4-cliques (one for every pair of disconnected directed edges other than the root).

\section{Conclusion}\label{Concl}
The underlying heuristic of powerful methods to count small rooted subgraphs is to enumerate copies in the neighborhood of each vertex or edge~\cite{Przulj2007,solava2012graphlet,Ali2014alignment,hocevar2014}. However, the power-law degrees of real-world graphs makes the average distance between vertices so small~\cite{watts1998small}, that at least one vertex neighborhood is very likely to span a positive fraction of the whole graph~\cite{chatterjee2009}. Therefore, the neighborhood approach ceases to become efficient to count medium-sized subgraphs.

Unfortunately, as is becoming increasingly obvious when considering for instance applications in biology, counting larger rooted subgraphs is needed to produce finer classifications of network data~\cite{solava2012graphlet,Ali2014alignment}. Furthermore, the need for counts of larger subgraphs aligns with results from random graph theory and statistics: for a general class of random graph models (termed inhomogeneous random graphs), larger subgraphs present smaller variances, and therefore are more powerful statistics to classify graphs~\cite{BickelLevina2012}. This suggests that counts of larger and larger subgraphs will be needed, and that a general understanding of medium-sized subgraph counting must be achieved.

To address this issue, and count medium-sized subgraphs efficiently, we first generalize it by introducing {\em connection matrices}. Connection matrices generalize most operators used alongside graphs: the adjacency matrix, the incidence matrix, the non-backtracking matrix, etc. For a rooted graph $F$ and a graph $G$, the connection matrix $\kappa(F,G)$ counts the number of copies of $F$ which connects any pair of possible locations for the roots of $F$ in $G$.

Then, we prove the surprising result that despite their generality, any connection matrices can be computed using basic matrix operations on the adjacency matrix. More precisely, with $A(\cdot)$ the map that associates to a graph its adjacency matrix, we show that the span of $A(\cdot)$ under classical matrix operations (operator, entry-wise and Kronecker products, transpose and vectorization) contains all connection matrices; i.e., $A(\cdot)$ is a generator of the set of connection matrices under classical matrix operations.

The surprising part of this result is not that all rooted counts can be retrieved from the adjacency matrix---this is expected. Rather, it is that in all generality, basic matrix operations are sufficient to count copies of any rooted graph in any graph. It proves that the algorithmic challenges raised by the subgraph counting problem are subsumed by the ones raised by the product or large matrices.

Finally, this reduction of the subgraph counting problems yields powerful algorithms. Indeed, the key advantage of using matrix products over neighborhood search is that taking one additional matrix product does not significantly increase the complexity. We leverage this by providing algorithms recovering the complexity, or indeed improving on, best known methods to count medium-sized rooted graphs. These algorithms substantiate our claim that using $A(\cdot)$ and classical matrix operations to compute any connection matrix becomes algorithmically efficient for large enough subgraphs.

\section{Appendix}
We use the following notations:
\vspace{-.3333\baselineskip}
\begin{itemize}[leftmargin= *]
\item[--] For two matrices $A$ and $B$, we write: $AB$ for the product of $A$ and $B$, $A\cdot B$ for the entry-wise (or Hadamard) product, $A\times B$ for the Kronecker product, $\vect (A)$ for the vectorization of $A$ and $A^\top$ for the transpose of $A$;
\vspace{-.5\baselineskip}
\item[--]  We write $[k] = \{1,2,\dots,k\}$ and for a set $X$ we write $|X|$ for the cardinality of $X$ and $X^k$ for the set of ordered $k$-tuples of elements of $X$. For a tuple $r$ we write $|r|$ for the length of $r$;
\vspace{-.5\baselineskip}
\item[--]  We denote ${\bf 1}_k$ the column vector of all 1 of dimension $k$.
\end{itemize}
\subsection{Proof of Theorem~\protect\ref{thm1}}
Before we proceed to prove Proposition~\ref{hp}, we introduce the following definition and lemma.
\begin{Definition}[Sub-tuple]
Set $r=r_1\cdots r_{|r|}$ an ordered tuple and $l\leq |r|$. Then, for each $k=k_1\cdots k_l\in [|r|]^l$ without repetition---i.e., such that $i\neq j\Rightarrow k_i\neq k_j$---we write $r_{|k}= r_{k_1}\cdots r_{k_l}$.
\end{Definition}
\begin{Lemma}
\label{lem}
Fix a graph $G$ of order $n$ and a rooted graph $F=[F,r,s]$. Then:
\vspace{-.5\baselineskip}
\begin{align*}
&\text{In all cases,}&&
\kappa([F,s,r],G)=\kappa(F,G)^\top,\ 
\kappa([F,rs,\emptyset],G)=\vect(\kappa(F,G)),\\
&\text{If $t= r_{|k}$,}&&
\kappa([F,t,s],G)=(\mathbb{I}\{i=j_{|k}\})_{i\in[n]^{|t|},j\in[n]^{|r|}}\kappa(F,G),\qquad\\
&\text{If $s=r_{|k}$,}&&
\kappa(F,G)=(\mathbb{I}\{i=j_{|k}\})_{i\in[n]^{|r|},j\in[n]^{|s|}}\cdot ({\bf 1}_{n^{|s|}}^\top\times \kappa([F,r,\emptyset],G)).\qquad
\end{align*}
\end{Lemma}
\begin{proof}
Recall from Definition~\ref{count} that $[F',j,i]\equiv [F,s,r]$ if there exists a bijective map $\phi:V(F)\to V(F')$ such that $\phi(r)=\phi(i)$, $\phi(s)=\phi(j)$ and $\phi(u)\phi(v)\in E(F') \Leftrightarrow uv\in E(F)$. In the following, we need to repeatedly manipulate $\phi$. Therefore, to simplify notation, we will write within this proof $[F',j,i]\equiv_\phi [F,s,r]$. We now prove each claim in succession. 
\vspace{-.3333\baselineskip}
\begin{itemize}[leftmargin= *]
\item[--] We show that $\kappa([F,s,r],G)=\kappa(F,G)^\top$. Fix $i\in[n]^{|r|}$ and $j\in[n]^{|s|}$. First, we directly observe that $[F',j,i]\equiv_\phi [F,s,r] \Leftrightarrow [F',i,j]\equiv_\phi [F,r,s]$. Now, from Definition~\ref{count}, we therefore have
\begin{align*}
\kappa([F,s,r],G)_{ji} &= \#\left\{[F',j,i]\subset [G,j,i]\,:\,[F',j,i]\equiv [F,s,r]\right\}\\
&= \#\left\{[F',i,j]\subset [G,i,j]\,:\,[F',i,j]\equiv [F,r,s]\right\}\\
& = \kappa([F,s,r],G)_{ij},
\end{align*}
which yields the claim.
\item[--] We show that  $\kappa([F,rs,\emptyset],G)=\vect(\kappa(F,G))$. Fix $i\in[n]^{|r|}$ and $j\in[n]^{|s|}$. First, from Definition~\ref{count}, we directly observe that $[F',i,j]\equiv_\phi [F,r,s] \Leftrightarrow [F',ij,\emptyset]\equiv_\phi [F,rs,\emptyset]$. Now, from Definition~\ref{count}, we therefore have
\begin{align*}
\vect(\kappa(F,G))_{ij}
&= \#\left\{[F',i,j]\subset [G,i,j]\,:\,[F',i,j]\equiv [F,r,s]\right\}\\
&= \#\left\{[F',ij,\emptyset]\subset [G,ij,\emptyset]\,:\,[F',ij,\emptyset]\equiv [F,rs,\emptyset]\right\}\\
&= \kappa([F,rs,\emptyset],G)_{ij},
\end{align*}
which yields the claim.
\item[--] Set $t= r_{|k}$. We show that $\kappa([F,t,s],G)\!=\!(\mathbb{I}\{i=w_{|k}\})_{i\in[n]^{|t|},w\in[n]^{|r|}}\kappa(F,G)$. Fix $i\in[n]^{|t|}$ and $j\in[n]^{|s|}$. There we observe that if $[F',i,j]\equiv_\phi [F,t,s]$, then i) for any tuple $w$ in $V(F)$, $[F',i\phi(w),j]\equiv_\phi [F,tw,s]$; ii) for any tuple $w$ of length at most $|t|$ in $[|t|]$ without repetition, we have that $[F',i_{|w},j]\equiv_\phi [F,t_{|w},s]$. Together, this shows by inclusion exclusion, that
\begin{align*}
\kappa([F,t,s],G)_{ij}
&= \#\big\{[F',i,j]\subset [G,i,j]\,:\,[F',i,j]\equiv [F,t,s]\big\}\\
&=\sum_{w\in[n]^{|r|}: w_{|k}=i}
\#\big\{[F',w,j]\subset [G,w,j]\,:\,[F',w,j]\equiv [F,r,s]\big\}\\
&=\sum_{w\in[n]^{|r|}: w_{|k}=i}\kappa(F,G)_{wj}.
\end{align*}
Therefore,
\begin{align*}
\big((\mathbb{I}\{i=w_{|k}\})_{i\in[n]^{|t|},w\in[n]^{|r|}}&\kappa(F,G)\big)_{ij}
=\sum_{w\in[n]^{|r|}}\mathbb{I}\{i=w_{|k}\}\kappa(F,G)_{wj}\\
&=\sum_{w\in[n]^{|r|}: w_{|k}=i}\kappa(F,G)_{wj}\\
&=\kappa([F,t,s],G)_{ij},
\end{align*}
which yields the claim.
\item[--] Set $s=r_{|k}$. We show that $\kappa(F,G)=(\mathbb{I}\{p=q_{|k}\})_{pq}\cdot({\bf 1}_{n^{|s|}}^\top\times\kappa([F,r,\emptyset],G))$, for $p$ in $[n]^{|r|}$ and $q$ in $[n]^{|s|}$. Fix $i\in[n]^{|r|}$ and $j\in[n]^{|s|}$. First, we observe that,
\[
\big((\mathbb{I}\{p=q_{|k}\})_{pq}\cdot({\bf 1}_{n^{|s|}}^\top\times\kappa([F,r,\emptyset],G))\big)_{ij} = 
\mathbb{I}\{i=j_{|k}\}\kappa([F,r,\emptyset],G)_{i}.
\]
Then, we have the following two implications:
\vspace{-.3333\baselineskip}
\begin{itemize}
\item[i)] If $[F',i,j]\equiv_\phi [F,r,s]$ then $[F',i,\emptyset]\equiv_\phi [F,r,\emptyset]$,
\item[ii)] If $[F',i,\emptyset]\equiv_\phi [F,r,\emptyset]$ then $[F',i,i_{|k}]\equiv_\phi [F,r,r_{|k}]$.
\end{itemize}
\vspace{-.3333\baselineskip}
Therefore, by inclusion exclusion, $\kappa([F,r,\emptyset],G)_{i}=\kappa([F,r,r_{|k}],G)_{ii_{|k}}$, and recalling that $s=r_{|k}$, we obtain
\[
\mathbb{I}\{i=j_{|k}\}\kappa([F,r,\emptyset],G)_{i}
 = \mathbb{I}\{i=j_{|k}\}\kappa([F,r,r_{|k}],G)_{ii_{|k}}
 = \kappa([F,r,s],G)_{ij},
\]
which yields the claim.\qedhere
\end{itemize}
\end{proof}
We can now prove Proposition~\ref{hp}.
\begin{proof}[Proof of Proposition~\protect\ref{hp}]
We prove each item in succession.
\vspace{-.3333\baselineskip}
\begin{enumerate}[leftmargin= *]
\item Fix $r',s'\subset r\cup s$. We show that $[F,r',s']\in\digamma$ using Lemma~\ref{lem}. First, by Lemma~\ref{lem}.1, $[F,rs,\emptyset]\in\digamma$, with $f_{[F,rs,\emptyset]}=\vect(f_F)$. Then, by Lemma~\ref{lem}.3, $[F,rs,s']\in\digamma$ since there exists a $k$ such that $s'=(rs)_{|k}$. Finally, by Lemma~\ref{lem}.2, $[F,r',s']\in\digamma$ since there exists a $k'$ such that $r'=(rs)_{|k'}$.
\item Assume that $r\cup s =V(F)$. Fix a graph $G$ and $i,j$ in $[|G|]^{|r|}$ and $[|G|]^{|s|}$ respectively. We first show that $\kappa(F,G)_{ij}\in\{0,1\}$. Assume that there exist $[F_1,i,j]$ and $[F_2,i,j]$, two copies of $F$ in $G$ at $i,j$. Let $\phi_1$ and $\phi_2$ be the associated adjacency preserving bijections mapping $F$ onto $F_1$ and $F_2$ respectively. Then---following Definition~\ref{count}---for all $p=(r\cup s)_{|k}$, $\phi_1(p)=\phi_2(p)=(r\cup s)_{|k}$ and therefore $\phi_1(p)=\phi_2(p)$ for all $p\in r\cup s=V(F)$, which directly yields that $F_1=F_2$. Hence, each entry of $\kappa(F,G)$ is at most 1.\\
Therefore, if $\kappa(F,G)\neq0$, we can call $F_{ij}$ the copy of $F$ in $G$ at $i,j$, and let $F_{ij}$ be the empty graph otherwise. Then, using that $rs\in V(F)\cap V(F')$, we have that
\begin{align*}
\big(\kappa(F,G)&\cdot\kappa(F',G)\big)_{ij}
=\mathbb{I}\{\kappa(F,G)_{ij}\neq0\}\kappa(F',G)_{ij}\\
&=\mathbb{I}\{\kappa(F,G)_{ij}\neq0\}
\#\big\{[\bar F,i,j]\subset [G,i,j]\,:\,[\bar F,i,j]\equiv [F',r,s]\big\}\\
&=\mathbb{I}\{\kappa(F,G)_{ij}\neq0\}\sum_{[\bar F,i,j]\subset [G,i,j]}\mathbb{I}\{[\bar F,i,j]\equiv [F',r,s]\}.\\
\intertext{Now we use twice the unicity of $F_{ij}$ to simplify the summation as follows:}
&=\sum_{[\bar F,i,j]\subset [G,i,j]}\mathbb{I}\{[\bar F\cup F_{ij},i,j]\equiv [F'\cup F,r,s]\}\\
&=\#\big\{[\hat F,i,j]\subset [G,i,j]\,:\,[\hat F,i,j]\equiv [F'\cup F,r,s]\big\}\\
&=\kappa([F\cup F',r,s],G)_{ij}.
\end{align*}
Therefore $\kappa([F\cup F',r,s],G) = \kappa(F,G)\cdot\kappa(F',G)$, so that $[F\cup F',r,s]\in\digamma$ with $f_{[F\cup F',r,s]} = f_{F}\cdot f_{F'}$.
\item Assume that $r\cup s =V(F)$ and $s\cup t =V(F'')$. Fix a graph $G$ and $i,j$ in $[|G|]^{|r|}$ and $[|G|]^{|t|}$ respectively. Then, as above, we have that each entry of $\kappa(F,G)$ and $\kappa(F'',G)$ are in $\{0,1\}$. Furthermore, 
\begin{align*}
\big(\kappa(F,G)&\kappa(F'',G)\big)_{ij}
= \sum_{w\in[|G|]^{|s|}}\kappa(F,G)_{iw}\kappa(F'',G)_{wj}\\
&= \sum_{w\in[|G|]^{|s|}}
        \sum_{[\bar F,i,w]\subset [G,i,w]}
        \sum_{[\bar F'',w,j]\subset [G,w,j]}
        \mathbb{I}\left\{
                \parbox{.25\textwidth}{\centering
                $[\bar F,i,w]\equiv [F,r,s]$
                \\[.1\baselineskip]
                $[\bar F'',w,j]\equiv [F,s,t]$}\right\}.\\
\intertext{There we observe that for each $[\hat F,iw,j]\equiv [F\cup F'',rs,t]$, there is a finite number of $\bar F, \bar F''$ such that $[\bar F,i,w]\equiv [F,r,s]$ and $[\bar F'',w,j]\equiv [F,s,t]$ and $\hat F = \bar F\cup \bar F''$. This number is independent of $w$ and $G$ and only depends on $F$ and $F''$; we call it $c_{F,F''}$. Resuming, we obtain:}
&= c_{F,F''}\sum_{w\in[|G|]^{|s|}}
\sum_{[\hat F,iw,j]\subset [G,iw,j]}
        \mathbb{I}\left\{[\hat F,iw,j]\equiv [F\cup F'',rs,t]\right\}\\
&= c_{F,F''}\sum_{[\hat F,i,j]\subset [G,i,j]}
        \mathbb{I}\left\{[\hat F,i,j]\equiv [F\cup F'',r,t]\right\}\\
&= c_{F,F''}\kappa([F\cup F'',r,t],G)_{ij}.
\end{align*}
Therefore $\kappa([F\cup F'',r,t],G) = c_{F,F''}^{-1}\kappa(F,G)\kappa(F'',G)$, so that $[F\cup F'',r,t]\in\digamma$ with $f_{[F\cup F'',r,t]} = c_{F,F''}^{-1}f_{F}f_{F''}$.\qedhere
\end{enumerate}
\end{proof}
With Proposition~\ref{hp} now proved, we may now turn to the proof of Theorem~\ref{thm1}.
\begin{proof}[Proof of Theorem~\protect\ref{thm1}]
As explained in the main body of this document, the proof consists in first deconstructing $F$ and then reconstructing it using Proposition~\ref{hp}. Resuming from p4:\\

\noindent (i) -- Let $t$ be any vertex in $V(F)$ and call $F'$ the subgraph of $F$ obtained after removing $t$. Then, for any $s\in V(F)\setminus\{t\}$, both $[F',V(F'),s]$ and $[(\{s,t\},\emptyset),s,t]$ are in $\digamma$ by the induction assumption. Thus, by Proposition~\ref{hp}.2,
\[
[F'\cup (\{s,t\},\emptyset),V(F'),t]=[(V(F),E(F')),V(F'),t]\in\digamma,
\]
and by Proposition~\ref{hp}.1, $[(V(F),E(F')),V(F),t]\in\digamma$.\\

\noindent (ii) -- Now call $F''$ the subgraph of $F$ with the same vertex set of $F$ but containing only the edges connected to $t$. Then, by the induction assumption, both $\bar F'' = [(V(F)\setminus\{s\},E(F'')\setminus\{st,ts\}),V(F)\setminus\{s\},t]$ and $[(\{t,s\},E(F'')\cap \{st,ts\}),t,s]$ are in $\digamma$. Thus, by Proposition~\ref{hp}.2, 
\[
[\bar F''\cup(\{t,s\},E(F'')\cap \{st,ts\}),V(F)\setminus\{s\},s]=[F'',V(F)\setminus\{s\},s]\in\digamma.
\]
Now, by Proposition~\ref{hp}.1, $[F'',V(F),t]\in\digamma$ and by Proposition~\ref{hp}.3,
\[
[(V(F),E(F'))\cup F'',V(F),t]=[F,V(F),t]\in\digamma.
\]
By Proposition~\ref{hp}.3 we then have $F\in\digamma$, which completes the proof.
\end{proof}

\subsection{Proof of Theorem~\protect\ref{thm2}}
We now list the formulas. See Figures~\ref{F4e},~\ref{F4n},~\ref{F5e},~\ref{F5n} and~\ref{FW9} for the rooted subgraphs' names. Code implementing the formulas in {\tt R}~\cite{Rsoftware} is available from the author on request. Each formula was verified to recover the true total count on a range of random graphs of varied sizes. Our benchmark was the LAD graph isomorphism counting algorithm~\cite{solnon2010all} as implemented in the {\tt R} {\tt igraph} package~\cite{igraph}.

We use the following notations:
\begin{itemize}
\item We write $\mathcal{F}_k^\dumbbell$ the set of edge-rooted simple graphs over $k$ vertices. We write $\mathcal{F}_k^\bullet$ the set of vertex-rooted simple graphs over $k$ vertices. 
\item For $e$ an edge in $G$ we write $\bar e$ for the reversed edge; i.e., if $e=ij$,then $\bar e = ji$.
\item For $M$ a matrix indexed by the directed edges of $G$, we write $\overline M$ and $\underline M$ for the columns and row reversal of $M$ respectively; i.e., $\overline M_{ee'} = M_{ e \bar e'}$, $\underline M_{ee'} = M_{\bar e e'}$ and $\underline M = \overline{M^\top}^\top$. Both operations can be performed linearly with the number of edges in $G$.
\item For $M$ a matrix, we write $\binom{M}{k}$ for the matrix of the same size as $M$ such that for each $p,q$, $\binom{M}{k}_{pq} = \binom{M_{pq}}{k}$. 
\item For $x$ a vector indexed by the directed edges of $G$, we write $\gamma(x)$ the vector indexed by the vertices of $G$ such that 
\[
\gamma(x) = \left(\sum_{j:ij\in E(G)} x_{ij}\right)
_{i\in V(G)}.
\]
\item We set $B^\vartriangle = B^2\cdot B^\top$, $B^\sqcap = B^2-B^\vartriangle$, $B^\boxempty = B^2\cdot (B^2)^\top$, $B^\boxtimes = B^\boxempty\cdot\overline{B^\boxempty}$, $B^\arro = B^2\cdot \overline{B^2}$, $B^\pfour = B^3\cdot(1- B^\top)\cdot(1-\overline{B})\cdot(1-\underline{B})$ and $B^{\Join} = B^\arro \underline{B^\vartriangle}$.
\end{itemize}
\paragraph{$\mathcal{F}_3^\dumbbell$:}
\begin{align*}
f_{P_3^\dumbbell}: B&\mapsto B{\bf 1}\\
f_{[(\{1,2,3\},\{12,23\}),21]}: B&\mapsto \underline{B{\bf 1}}\\
f_{C_3^\dumbbell}: B&\mapsto B^\vartriangle{\bf 1}
\end{align*}
\paragraph{$\mathcal{F}_3^\bullet$:}
\begin{align*}
f_{P_3^\bullet}: B&\mapsto \gamma\left(f_{P_3^\dumbbell}(B)\right)\\
f_{[(\{1,2,3\},\{12,23\}),2]}: B&\mapsto \frac{1}{2}\gamma\left(f_{[(\{1,2,3\},\{12,23\}),21]}(B)\right)\\
f_{C_3^\bullet}: B&\mapsto \frac{1}{2}\gamma\left(f_{C_3^\dumbbell}(B)\right)
\end{align*}
\paragraph{$\mathcal{F}_4^\dumbbell$:}
\begin{figure}
\centering
\includegraphics[width=.5\textwidth]{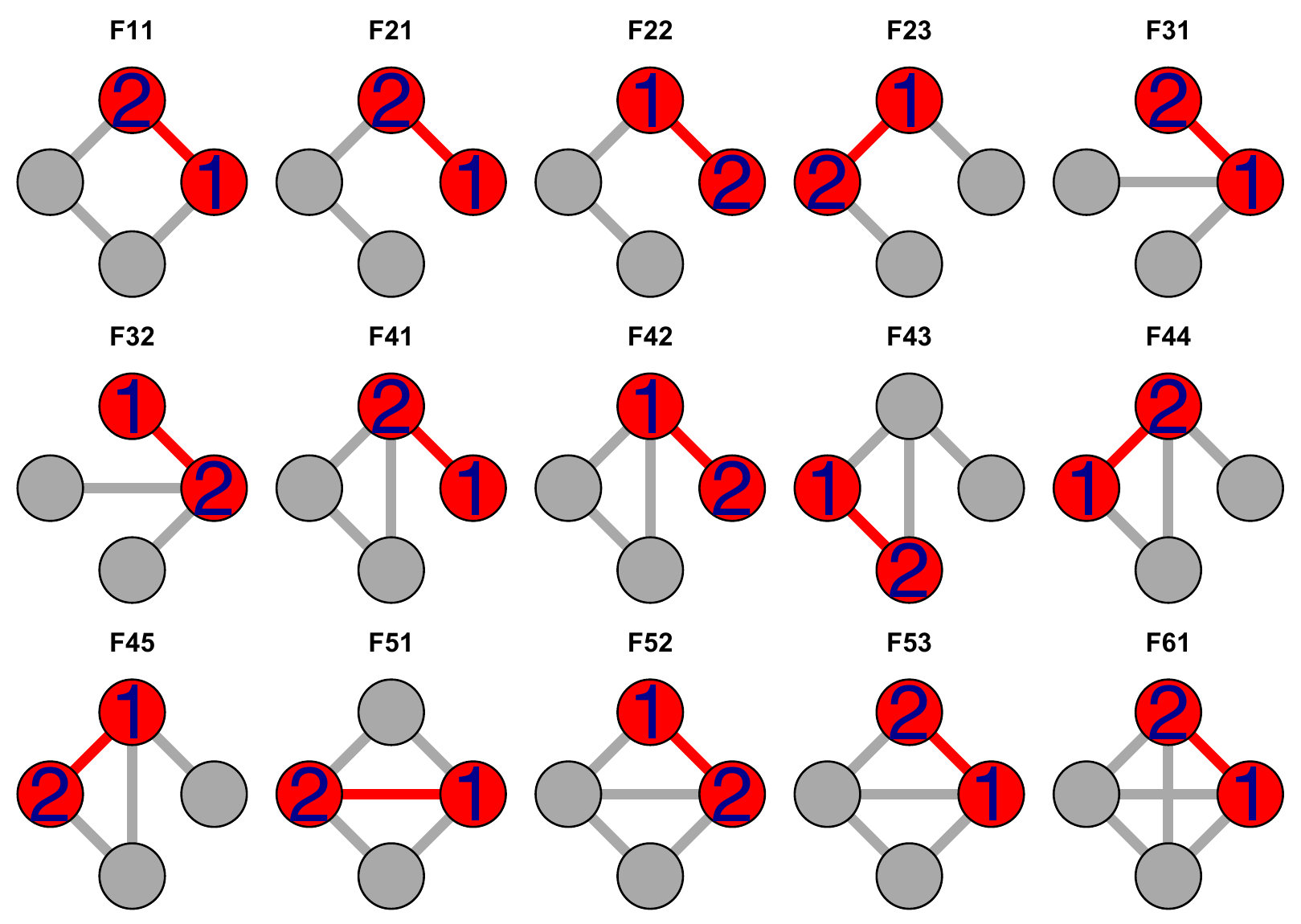}
\caption{\label{F4e} $\mathcal{F}_4^\dumbbell$, where each subgraph is rooted at the colored directed edge $12$. In the following formulas, we add ``$\,^\dumbbell\,$'' as exponent to distinguish with vertex rooted subgraphs.}
\end{figure}
\begin{align*}
f_{X11^\dumbbell}: B&\mapsto B^\boxempty{\bf 1}\\
&\\%
f_{X21^\dumbbell}: B&\mapsto B^\sqcap{\bf 1}\\
f_{X22^\dumbbell}: B&\mapsto \underline{f_{X21}(B)}\\
f_{X23^\dumbbell}: B&\mapsto (B{\bf 1})\cdot ({\bf 1}B)-B^\vartriangle{\bf 1}\\
&\\%
f_{X31^\dumbbell}: B&\mapsto \binom{B{\bf 1}}{2}\\
f_{X32^\dumbbell}: B&\mapsto \underline{f_{X31}(B)}\\
&\\%
f_{X41^\dumbbell}: B&\mapsto B^\arro{\bf 1}/2\\
f_{X42^\dumbbell}: B&\mapsto \underline{f_{X41}(B)}\\
f_{X43^\dumbbell}: B&\mapsto {\bf 1}B^\arro\\
f_{X44^\dumbbell}: B&\mapsto (B^\vartriangle){\bf 1}\cdot(B{\bf 1}-1)\\
f_{X45^\dumbbell}: B&\mapsto \underline{f_{X45}(B)}\\
&\\%
f_{X51^\dumbbell}: B&\mapsto \binom{B^\vartriangle{\bf 1}}{2}\\
f_{X52^\dumbbell}: B&\mapsto (B^\boxempty\cdot \overline{B^2}){\bf 1}\\
f_{X53^\dumbbell}: B&\mapsto \underline{f_{X52}(B)}\\
&\\%
f_{X61^\dumbbell}: B&\mapsto B^\boxtimes{\bf 1}/2
\end{align*}
\paragraph{$\mathcal{F}_4^\bullet$:}
\begin{figure}
\centering
\includegraphics[width=.5\textwidth]{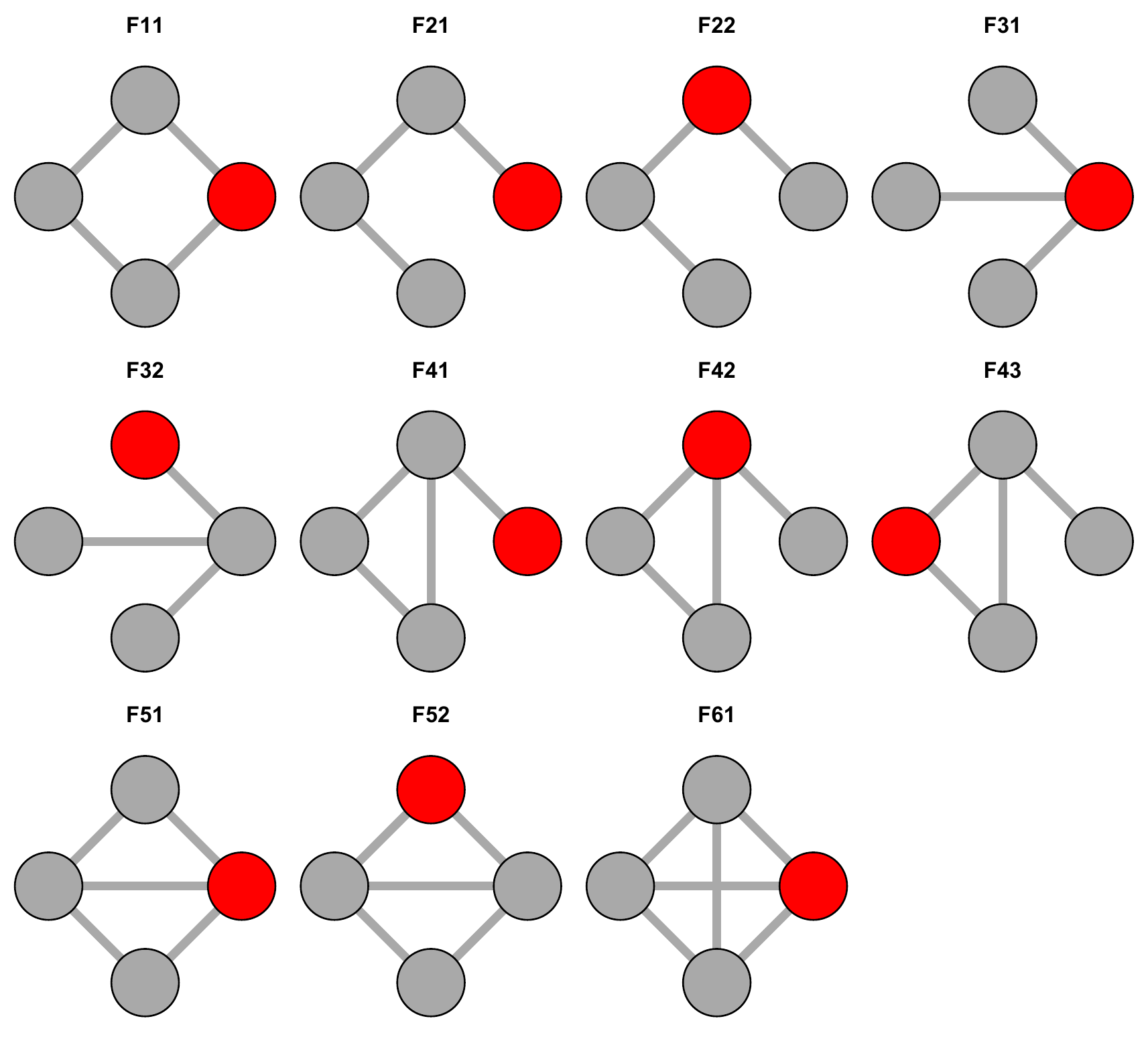}
\caption{\label{F4n} $\mathcal{F}_4^\bullet$, where each subgraph is rooted at the colored vertex. In the following formulas, we add ``$\,^\bullet\,$'' as exponent to distinguish with edge rooted subgraphs.}
\end{figure}
\begin{align*}
f_{X11^\bullet}: B&\mapsto
 \frac{1}{2}\gamma\left(f_{X11^\dumbbell}(B)\right)\\
f_{X21^\bullet}: B&\mapsto
\gamma\left(f_{X21^\dumbbell}(B)\right)\\
f_{X22^\bullet}: B&\mapsto
\gamma\left(f_{X22^\dumbbell}(B)\right)\\
f_{X31^\bullet}: B&\mapsto
 \frac{1}{3}\gamma\left(f_{X31^\dumbbell}(B)\right)\\
f_{X32^\bullet}: B&\mapsto
\gamma\left(f_{X32^\dumbbell}(B)\right)\\
f_{X41^\bullet}: B&\mapsto
\gamma\left(f_{X41^\dumbbell}(B)\right)\\
f_{X42^\bullet}: B&\mapsto
\gamma\left(f_{X42^\dumbbell}(B)\right)\\
f_{X43^\bullet}: B&\mapsto
\gamma\left(f_{X43^\dumbbell}(B)\right)\\
f_{X51^\bullet}: B&\mapsto
\gamma\left(f_{X51^\dumbbell}(B)\right)\\
f_{X52^\bullet}: B&\mapsto
\frac{1}{2}\gamma\left(f_{X52^\dumbbell}(B)\right)\\
f_{X61^\bullet}: B&\mapsto
\frac{1}{3}\gamma\left(f_{X61^\dumbbell}(B)\right)
\end{align*}
\paragraph{$\mathcal{F}_5^\dumbbell$:}
\begin{figure}
\centering
\includegraphics[width=\textwidth]{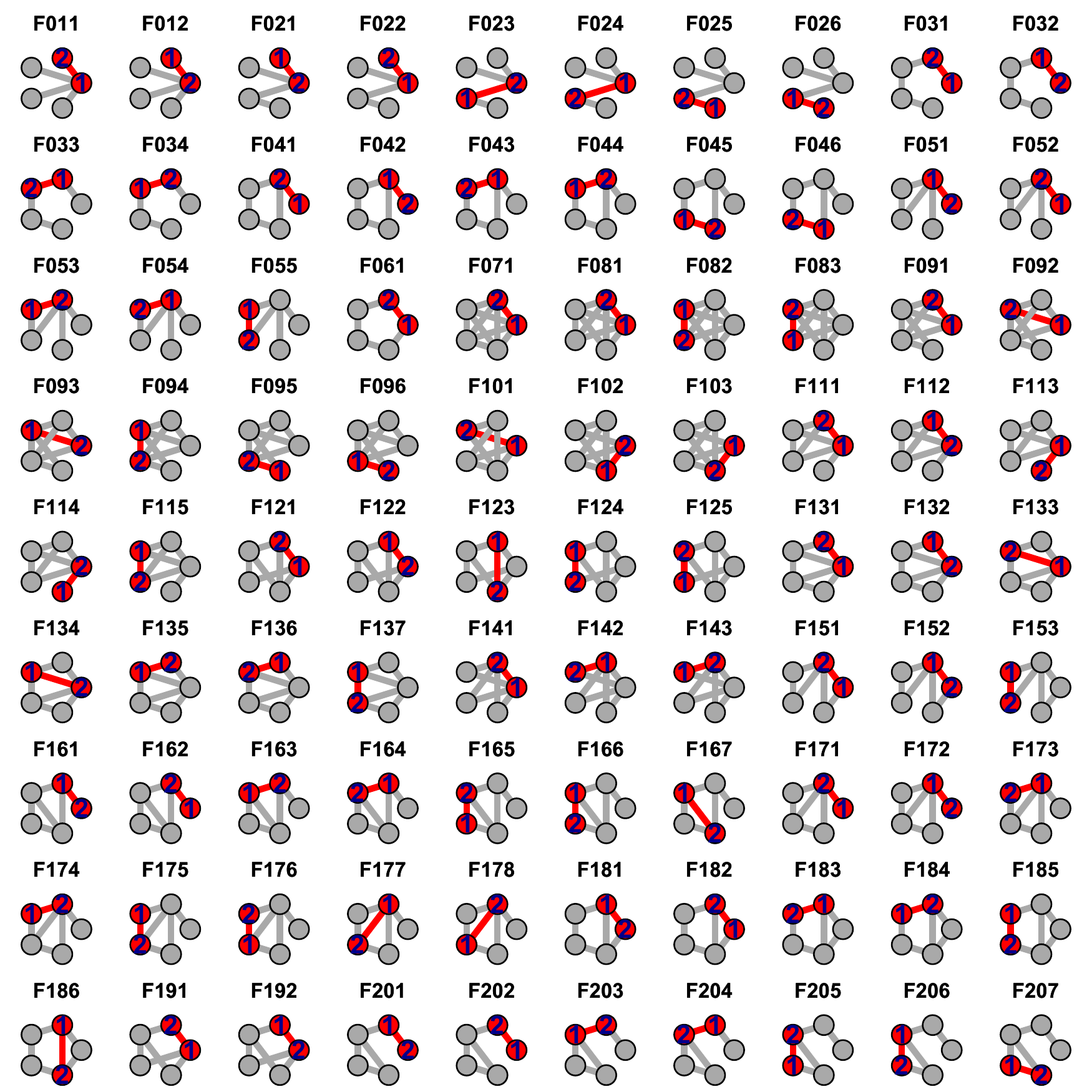}
\caption{\label{F5e} $\mathcal{F}_5^\dumbbell$, where each subgraph is rooted at the colored and directed edge $12$. In the following formulas, we add ``$\,^\dumbbell\,$'' as exponent to distinguish with vertex rooted subgraphs.}
\end{figure}
\begin{align*}
f_{X011^\dumbbell}: B&\mapsto \binom{\underline{B{\bf 1}}}{3}\\
f_{X012^\dumbbell}: B&\mapsto \underline{f_{X011}(B)}\\
&\\%
f_{X021^\dumbbell}: B&\mapsto (B^\sqcap{\bf 1})\cdot (B{\bf 1}-1)-B^\arro{\bf 1}\\
f_{X022^\dumbbell}: B&\mapsto \underline{f_{X021}(B)}\\
f_{X023^\dumbbell}: B&\mapsto \binom{B{\bf 1}}{2}\cdot (\underline{B{\bf 1}}) - (B^\vartriangle{\bf 1})\cdot (B{\bf 1}-1)\\
f_{X024^\dumbbell}: B&\mapsto \underline{f_{X023}(B)}\\
f_{X025^\dumbbell}: B&\mapsto B\binom{B{\bf 1}}{2}-B^\arro{\bf 1}\\
f_{X026^\dumbbell}: B&\mapsto \underline{f_{X025}(B)}\\
&\\%
f_{X031^\dumbbell}: B&\mapsto B^\pfour{\bf 1}\\
f_{X032^\dumbbell}: B&\mapsto \underline{f_{X031}(B)}\\
f_{X033^\dumbbell}: B&\mapsto (B^\sqcap{\bf 1})\cdot \underline{B{\bf 1}}-B^\boxempty{\bf 1}-B^\arro{\bf 1}\\
f_{X034^\dumbbell}: B&\mapsto \underline{f_{X033}(B)}\\
&\\%
f_{X041^\dumbbell}: B&\mapsto (B^\pfour\cdot\overline{B^2}){\bf 1}/2\\
f_{X042^\dumbbell}: B&\mapsto \underline{f_{X041}(B)}\\
f_{X043^\dumbbell}: B&\mapsto (B^\boxempty{\bf 1})\cdot (\underline{B{\bf 1}}-1)-(B^\boxempty\cdot\underline{B^2}){\bf 1}\\
f_{X044^\dumbbell}: B&\mapsto \underline{f_{X043}(B)}\\
f_{X045^\dumbbell}: B&\mapsto {\bf 1}(B^\pfour\cdot\overline{B^2})\\
f_{X046^\dumbbell}: B&\mapsto \underline{f_{X045}(B)}\\
&\\%
f_{X051^\dumbbell}: B&\mapsto ({\bf 1}B^\arro)\cdot (B{\bf 1}-2)/2\\
f_{X052^\dumbbell}: B&\mapsto \underline{f_{X052}(B)}\\
f_{X053^\dumbbell}: B&\mapsto (B^\vartriangle{\bf 1})\cdot\binom{B{\bf 1}-1}{2}\\
f_{X054^\dumbbell}: B&\mapsto \underline{f_{X053}(B)}\\
f_{X055^\dumbbell}: B&\mapsto \big(B^\arro\underline{B}-2B^\arro\big){\bf 1}/2\\
&\\%
f_{X061^\dumbbell}: B&\mapsto \big(B^3\cdot B^{2\top}\big){\bf 1}\\
&\\%
f_{X071^\dumbbell}: B&\mapsto \big(B^\boxtimes\cdot B^{\Join}\big){\bf 1}/6\\
&\\%
f_{X081^\dumbbell}: B&\mapsto \big(B^\boxtimes\cdot (B^\arro B)\big){\bf 1}/2\\
f_{X082^\dumbbell}: B&\mapsto {\bf 1}\big(B^\boxtimes\cdot (B^\arro B)\big)/2\\
f_{X083^\dumbbell}: B&\mapsto \underline{f_{X082}(B)}\\
&\\%
f_{X091^\dumbbell}: B&\mapsto \big(B^\boxtimes B^\vartriangle-2B^\boxtimes\big){\bf 1}/2\\
f_{X092^\dumbbell}: B&\mapsto \big((B^\boxtimes B^\boxempty)\cdot B \big){\bf 1}\\
f_{X093^\dumbbell}: B&\mapsto \underline{f_{X092}(B)}\\
f_{X094^\dumbbell}: B&\mapsto (B^\boxtimes{\bf 1})\cdot (B^\vartriangle{\bf 1}-2)/2\\
f_{X095^\dumbbell}: B&\mapsto \big({\bf 1}(B^\boxtimes B^\vartriangle)-2B^\boxtimes{\bf 1}\big)/2\\
f_{X096^\dumbbell}: B&\mapsto \underline{f_{X095}(B)}\\
&\\%
f_{X101^\dumbbell}: B&\mapsto \big(B^\boxempty\cdot B^{\Join}\big){\bf 1}\\
f_{X102^\dumbbell}: B&\mapsto \binom{(B^\vartriangle \underline{B^\vartriangle})\cdot B^\top}{2}{\bf 1}\\
f_{X103^\dumbbell}: B&\mapsto \underline{f_{X102}(B)}\\
&\\%
f_{X111^\dumbbell}: B&\mapsto (B^\boxtimes{\bf 1})\cdot (\underline{B{\bf1}}-2)/2\\
f_{X112^\dumbbell}: B&\mapsto \underline{f_{X111}(B)}\\
f_{X113^\dumbbell}: B&\mapsto {\bf 1}(B^\boxtimes B - 2B^\boxtimes)/6\\
f_{X114^\dumbbell}: B&\mapsto \underline{f_{X113}(B)}\\
f_{X115^\dumbbell}: B&\mapsto (B^\boxtimes B){\bf 1} - 2{\bf 1}B^\boxtimes\\
&\\%
f_{X121^\dumbbell}: B&\mapsto {\bf 1}\big((B^3-1)\cdot(B^\vartriangle\underline{B^\vartriangle})\cdot B^\top\big)\\
f_{X122^\dumbbell}: B&\mapsto \underline{f_{X121}(B)}\\
f_{X123^\dumbbell}: B&\mapsto {\bf 1}\big((B^\boxempty B^\vartriangle)\cdot B^2\big)/2\\
f_{X124^\dumbbell}: B&\mapsto {\bf 1}(B^\boxempty\cdot(B^\arro B))/2\\
f_{X125^\dumbbell}: B&\mapsto \underline{f_{X124}(B)}\\
&\\%
f_{X131^\dumbbell}: B&\mapsto ((B^\boxempty B^\boxempty)\cdot B^\top){\bf 1}\\
f_{X132^\dumbbell}: B&\mapsto \underline{f_{X131}(B)}\\
f_{X133^\dumbbell}: B&\mapsto ((B^\boxempty\cdot\underline{B^2}){\bf 1})\cdot(B^\vartriangle{\bf 1}-1)-B^\boxtimes{\bf 1}\\
f_{X134^\dumbbell}: B&\mapsto \underline{f_{X133}(B)}\\
f_{X135^\dumbbell}: B&\mapsto ((B^\boxempty\cdot\underline{B^2})B^\vartriangle-B^\boxempty\cdot\underline{B^2}-B^\boxtimes){\bf 1}\\
f_{X136^\dumbbell}: B&\mapsto \underline{f_{X135}(B)}\\
f_{X137^\dumbbell}: B&\mapsto ((B^3\cdot B^2\cdot B^\top)(B^\vartriangle-1)-B^\boxtimes){\bf 1}\\
&\\%
f_{X141^\dumbbell}: B&\mapsto \binom{B^\vartriangle{\bf 1}}{3}\\
f_{X142^\dumbbell}: B&\mapsto B^\vartriangle \binom{B^\vartriangle{\bf 1}-1}{2}\\
f_{X143^\dumbbell}: B&\mapsto \underline{f_{X142}(B)}\\
&\\%
f_{X151^\dumbbell}: B&\mapsto {\bf 1}((B^\arro B^\vartriangle)\cdot(1-B)\cdot(1-\underline{B}))/2\\
f_{X152^\dumbbell}: B&\mapsto \underline{f_{X151}(B)}\\
f_{X153^\dumbbell}: B&\mapsto ((B^\arro B^\vartriangle)\cdot(1-B)\cdot(1-\underline{B})){\bf 1}/2\\
&\\%
f_{X161^\dumbbell}: B&\mapsto {\bf 1}((B^\boxempty\cdot\overline{B^2})\underline{B}-B^\boxtimes-B^\boxempty\cdot\overline{B^2})/2\\
f_{X162^\dumbbell}: B&\mapsto \underline{f_{X161}(B)}\\
f_{X163^\dumbbell}: B&\mapsto ((B^\boxempty\cdot\underline{B^2}){\bf 1})\cdot(B{\bf 1}-1)-B^\boxtimes){\bf 1}\\
f_{X164^\dumbbell}: B&\mapsto \underline{f_{X163}(B)}\\
f_{X165^\dumbbell}: B&\mapsto (((B^\boxempty\cdot\overline{B^2})\underline{B})\cdot(1-\overline{B})-B^\boxtimes){\bf 1}\\
f_{X166^\dumbbell}: B&\mapsto \underline{f_{X165}(B)}\\
f_{X167^\dumbbell}: B&\mapsto (B^\vartriangle{\bf 1}-1)\cdot (B^\arro{\bf 1})-B^\boxtimes{\bf 1}\\
&\\%
f_{X171^\dumbbell}: B&\mapsto B \binom{B^\vartriangle{\bf 1}}{2}-(B^\boxempty\cdot\overline{B^2}){\bf 1}\\
f_{X172^\dumbbell}: B&\mapsto \underline{f_{X171}(B)}\\
f_{X173^\dumbbell}: B&\mapsto ((B^\boxempty\cdot\overline{B^2}){\bf 1})\cdot(B{\bf 1}-2)\\
f_{X174^\dumbbell}: B&\mapsto \underline{f_{X173}(B)}\\
f_{X175^\dumbbell}: B&\mapsto (((B^\boxempty\cdot\overline{B^2})B)\cdot(1-\overline{B})\cdot(1-B^\top)){\bf 1}\\
f_{X176^\dumbbell}: B&\mapsto \underline{f_{X175}(B)}\\
f_{X177^\dumbbell}: B&\mapsto \binom{B^\vartriangle{\bf 1}}{2}\cdot(B{\bf 1}-2)\\
f_{X178^\dumbbell}: B&\mapsto \underline{f_{X177}(B)}\\
&\\%
f_{X181^\dumbbell}: B&\mapsto ((B^\vartriangle B^\boxempty)\cdot(1-\overline{B})\cdot(1-B)){\bf 1}\\
f_{X182^\dumbbell}: B&\mapsto \underline{f_{X181}(B)}\\
f_{X183^\dumbbell}: B&\mapsto (((B^3\cdot B^\top)B^\vartriangle)\cdot(1-B)-B^\boxempty\cdot\underline{B^2}){\bf 1}\\
f_{X184^\dumbbell}: B&\mapsto \underline{f_{X183}(B)}\\
f_{X185^\dumbbell}: B&\mapsto {\bf 1}((B^\vartriangle B^\boxempty)\cdot(1-\overline{B})\cdot(1-B))\\
f_{X186^\dumbbell}: B&\mapsto (B^\vartriangle{\bf 1})\cdot (B^\boxempty{\bf 1})-(B^\boxempty\cdot(\overline{B^2}+\underline{B^2})){\bf 1}\\
&\\%
f_{X191^\dumbbell}: B&\mapsto {\bf 1}\binom{B^3\cdot B^\top}{2}\\
f_{X192^\dumbbell}(B) &=\underline{f_{X191}(B)}\\
&\\%
f_{X201^\dumbbell}: B&\mapsto {\bf 1}((B^\arro B)\cdot(1-B)\cdot(1-\underline{B})\cdot(1-\overline{B})\cdot(1-B^\top))/2\\
f_{X202^\dumbbell}: B&\mapsto \underline{f_{X201}(B)}\\
f_{X203^\dumbbell}: B&\mapsto ({\bf 1}B^\arro)\cdot (B{\bf 1}/2)-(B^\boxempty\cdot\underline{B^2}){\bf 1}\\
f_{X204^\dumbbell}: B&\mapsto \underline{f_{X203}(B)}\\
f_{X205^\dumbbell}: B&\mapsto (B^\vartriangle{\bf 1})\cdot (B^\sqcap{\bf 1})-(B^\arro-B^\boxempty\cdot\overline{B^2}){\bf 1}\\
f_{X206^\dumbbell}: B&\mapsto \underline{f_{X205}(B)}\\
f_{X207^\dumbbell}: B&\mapsto ((B^\arro B)\cdot(1-B)\cdot(1-\underline{B})\cdot(1-\overline{B})\cdot(1-B^\top)){\bf 1}
\end{align*}
\paragraph{$\mathcal{F}_5^\bullet$:}
\begin{figure}
\centering
\includegraphics[width=\textwidth]{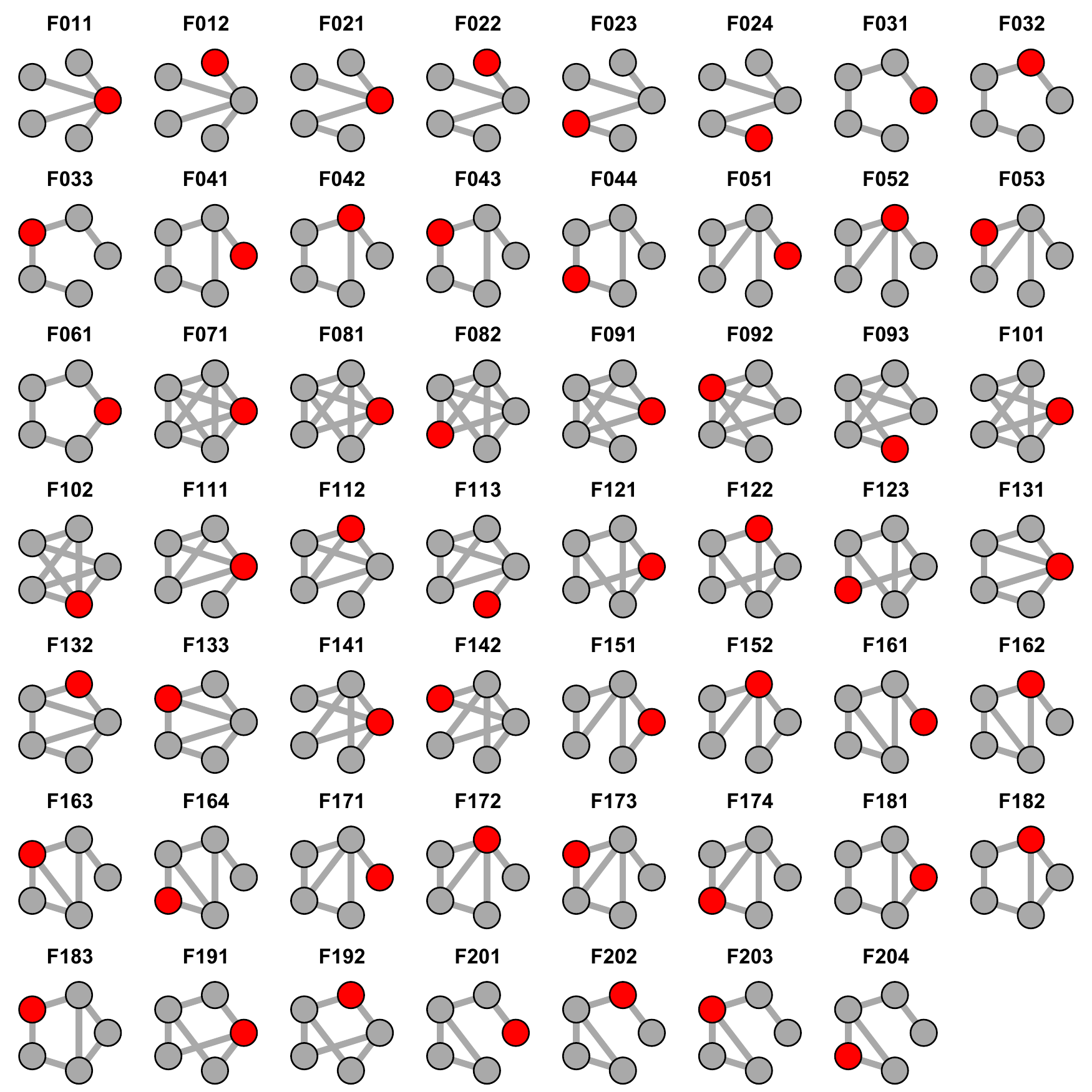}
\caption{\label{F5n} $\mathcal{F}_5^\bullet$, where each subgraph is rooted at the colored vertex. In the following formulas, we add ``$\,^\bullet\,$'' as exponent to distinguish with edge rooted subgraphs.}
\end{figure}
\begin{align*}
f_{X011^\bullet}: B&\mapsto
 \frac{1}{4}\gamma\left(f_{X011^\dumbbell}(B)\right)\\
f_{X012^\bullet}: B&\mapsto
 \gamma\left(f_{X012^\dumbbell}(B)\right)\\&\\
f_{X021^\bullet}: B&\mapsto
 \frac{1}{2}\gamma\left(f_{X022^\dumbbell}(B)\right)\\
f_{X022^\bullet}: B&\mapsto
 \gamma\left(f_{X021^\dumbbell}(B)\right)\\
f_{X023^\bullet}: B&\mapsto
 \gamma\left(f_{X023^\dumbbell}(B)\right)\\
f_{X024^\bullet}: B&\mapsto
 \gamma\left(f_{X025^\dumbbell}(B)\right)\\&\\
f_{X031^\bullet}: B&\mapsto
 \gamma\left(f_{X031^\dumbbell}(B)\right)\\
f_{X032^\bullet}: B&\mapsto
 \gamma\left(f_{X032^\dumbbell}(B)\right)\\
f_{X033^\bullet}: B&\mapsto
 \frac{1}{2}\gamma\left(f_{X034^\dumbbell}(B)\right)\\&\\
f_{X041^\bullet}: B&\mapsto
 \gamma\left(f_{X041^\dumbbell}(B)\right)\\
f_{X042^\bullet}: B&\mapsto
 \frac{1}{2}\gamma\left(f_{X043^\dumbbell}(B)\right)\\
f_{X043^\bullet}: B&\mapsto
 \gamma\left(f_{X044^\dumbbell}(B)\right)\\
f_{X044^\bullet}: B&\mapsto
 \frac{1}{2}\gamma\left(f_{X045^\dumbbell}(B)\right)\\&\\
f_{X051^\bullet}: B&\mapsto
 \gamma\left(f_{X052^\dumbbell}(B)\right)\\
f_{X052^\bullet}: B&\mapsto
 \frac{1}{2}\gamma\left(f_{X051^\dumbbell}(B)\right)\\
f_{X053^\bullet}: B&\mapsto
 \gamma\left(f_{X053^\dumbbell}(B)\right)\\&\\
f_{X061^\bullet}: B&\mapsto
 \frac{1}{2}\gamma\left(f_{X061^\dumbbell}(B)\right)\\&\\
f_{X071^\bullet}: B&\mapsto
 \frac{1}{4}\gamma\left(f_{X071^\dumbbell}(B)\right)\\&\\
f_{X081^\bullet}: B&\mapsto
 \frac{1}{2}\gamma\left(f_{X081^\dumbbell}(B)\right)\\
f_{X082^\bullet}: B&\mapsto
 \frac{1}{3}\gamma\left(f_{X083^\dumbbell}(B)\right)\\&\\
f_{X091^\bullet}: B&\mapsto
 \gamma\left(f_{X091^\dumbbell}(B)\right)\\
f_{X092^\bullet}: B&\mapsto
 \gamma\left(f_{X096^\dumbbell}(B)\right)\\
f_{X093^\bullet}: B&\mapsto
 \frac{1}{2}\gamma\left(f_{X095^\dumbbell}(B)\right)\\&\\
f_{X101^\bullet}: B&\mapsto
 \frac{1}{2}\gamma\left(f_{X101^\dumbbell}(B)\right)\\
f_{X102^\bullet}: B&\mapsto
 \frac{1}{4}\gamma\left(f_{X102^\dumbbell}(B)\right)\\&\\
f_{X111^\bullet}: B&\mapsto
 \frac{1}{3}\gamma\left(f_{X111^\dumbbell}(B)\right)\\
f_{X112^\bullet}: B&\mapsto
 \gamma\left(f_{X112^\dumbbell}(B)\right)\\
f_{X113^\bullet}: B&\mapsto
 \gamma\left(f_{X114^\dumbbell}(B)\right)\\&\\
f_{X121^\bullet}: B&\mapsto
 \frac{1}{2}\gamma\left(f_{X121^\dumbbell}(B)\right)\\
f_{X122^\bullet}: B&\mapsto
 \frac{1}{2}\gamma\left(f_{X122^\dumbbell}(B)\right)\\
f_{X123^\bullet}: B&\mapsto
 \frac{1}{2}\gamma\left(f_{X125^\dumbbell}(B)\right)\\&\\
f_{X131^\bullet}: B&\mapsto
 \frac{1}{2}\gamma\left(f_{X131^\dumbbell}(B)\right)\\
f_{X132^\bullet}: B&\mapsto
 \gamma\left(f_{X132^\dumbbell}(B)\right)\\
f_{X133^\bullet}: B&\mapsto
 \gamma\left(f_{X134^\dumbbell}(B)\right)\\&\\
f_{X141^\bullet}: B&\mapsto
 \gamma\left(f_{X141^\dumbbell}(B)\right)\\
f_{X142^\bullet}: B&\mapsto
 \frac{1}{2}\gamma\left(f_{X143^\dumbbell}(B)\right)\\&\\
f_{X151^\bullet}: B&\mapsto
 \gamma\left(f_{X151^\dumbbell}(B)\right)\\
f_{X152^\bullet}: B&\mapsto
 \frac{1}{4}\gamma\left(f_{X152^\dumbbell}(B)\right)\\&\\
f_{X161^\bullet}: B&\mapsto
 \gamma\left(f_{X162^\dumbbell}(B)\right)\\
f_{X162^\bullet}: B&\mapsto
 \gamma\left(f_{X161^\dumbbell}(B)\right)\\
f_{X163^\bullet}: B&\mapsto
 \gamma\left(f_{X163^\dumbbell}(B)\right)\\
f_{X164^\bullet}: B&\mapsto
 \frac{1}{2}\gamma\left(f_{X165^\dumbbell}(B)\right)\\&\\
f_{X171^\bullet}: B&\mapsto
 \gamma\left(f_{X171^\dumbbell}(B)\right)\\
f_{X172^\bullet}: B&\mapsto
 \gamma\left(f_{X172^\dumbbell}(B)\right)\\
f_{X173^\bullet}: B&\mapsto
 \gamma\left(f_{X175^\dumbbell}(B)\right)\\
f_{X174^\bullet}: B&\mapsto
 \frac{1}{2}\gamma\left(f_{X176^\dumbbell}(B)\right)\\&\\
f_{X181^\bullet}: B&\mapsto
 \frac{1}{2}\gamma\left(f_{X182^\dumbbell}(B)\right)\\
f_{X182^\bullet}: B&\mapsto
 \gamma\left(f_{X181^\dumbbell}(B)\right)\\
f_{X183^\bullet}: B&\mapsto
 \gamma\left(f_{X184^\dumbbell}(B)\right)\\&\\
f_{X191^\bullet}: B&\mapsto
 \frac{1}{3}\gamma\left(f_{X191^\dumbbell}(B)\right)\\
f_{X192^\bullet}: B&\mapsto
 \frac{1}{2}\gamma\left(f_{X192^\dumbbell}(B)\right)\\&\\
f_{X201^\bullet}: B&\mapsto
 \gamma\left(f_{X202^\dumbbell}(B)\right)\\
f_{X202^\bullet}: B&\mapsto
 \gamma\left(f_{X201^\dumbbell}(B)\right)\\
f_{X203^\bullet}: B&\mapsto
 \gamma\left(f_{X203^\dumbbell}(B)\right)\\
f_{X204^\bullet}: B&\mapsto
 \gamma\left(f_{X205^\dumbbell}(B)\right)
\end{align*}
\paragraph{$\mathcal{F}_6^\dumbbell$:}
In this section, we do not list all formulas as above: there are too many of them. Instead, we provide one general formula for all of them. 

Fix $[F,a_1a_2]\in \mathcal{F}_6^\dumbbell$. Let $\{b_1,b_2\},\{c_1,c_2\}$ be a partition of $V(F)\setminus\{a_1,a_2\}$, so that $V(F) = \{a_1,a_2,b_1,b_2,c_1,c_2\}$. Then, let $F_a$ be the subgraph of $F$ induced by $\{a_1,a_2,b_1,b_2\}$; i.e., $F_a = (\{a_1,a_2,b_1,b_2\},V(F)\cap \{a_1,a_2,b_1,b_2\}^2)$. In the same fashion, let $F_b$ be the subgraph of $F$ induced by $\{b_1,b_2,c_1,c_2\}$ and $F_c$ be the subgraph of $F$ induced by $\{c_1,c_2,a_1,a_2\}$.

Then, following Lemma~\ref{lem}, we have that for any simple graph $G$
\[
\kappa(F,G) = \Big(\big(\kappa([F_a,a_1a_2,b_1b_2],G)\kappa([F_b,b_1b_2,c_1c_2],G)\big)\cdot\kappa([F_c,a_1a_2,c_1c_2],G)\Big){\bf 1},
\]
and
\begin{equation}\label{6count}
f_F : B\mapsto \big(f_{[F_a,a_1a_2,b_1b_2]}(B)f_{[F_b,b_1b_2,c_1c_2]}(B)\big)\cdot f_{[F_c,a_1a_2,c_1c_2]}(B).
\end{equation}
Since $F_a, F_b$ and $F_c$ are order 4 rooted graphs, we expressed their connection matrix above, and~\eqref{6count} yields a formula for any edge-rooted order six simple graph.
\paragraph{$\mathcal{W}_9$:}
\begin{figure}
\centering
\includegraphics[width=.9\textwidth]{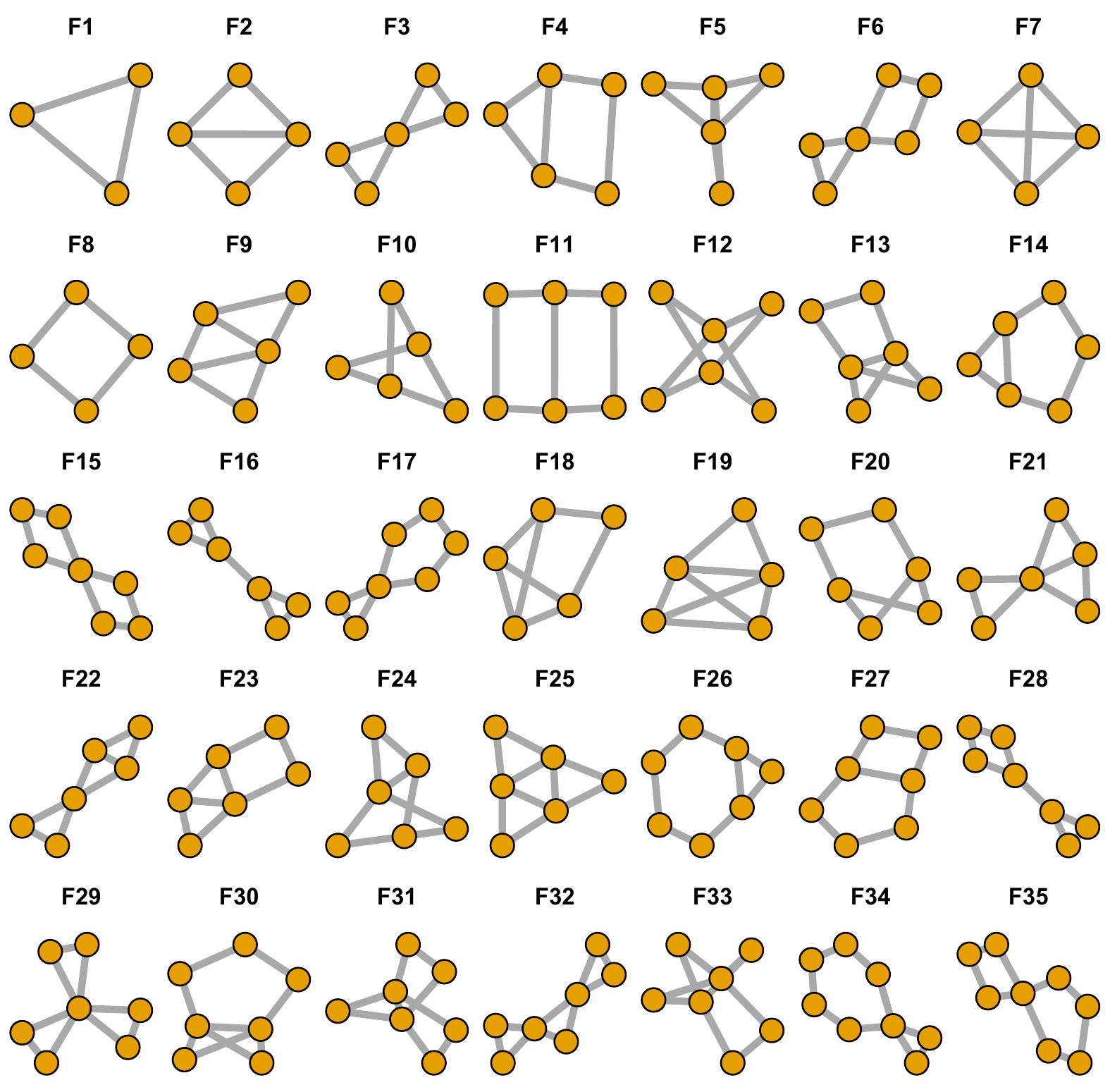}
\caption{\label{FW9} All subgraphs that can be induced by closed walks of length $k$ ($k\leq 9$) and that do not backtrack at any step.}
\end{figure}

We now turn to computing the total count of the $k$-cycles for $k\leq 9$. We use the exact same approach as in~\cite{harary71cycle,alon1997cycles,perepeshko2009cycle}, but instead of using the adjacency matrix $A$, we use the non-backtracking matrix $B$. Therefore, where their method requires to compute the number of copies of subgraphs that can be induced by all walks of length $k$, we need only compute the number of copies of subgraphs that can be induced by all closed walk of length $k$ such that visiting the walk twice in succession does not involve traversing the same edge twice in succession (walks that do not backtrack at any step, not event the first or the last). We list and name all these subgraphs in Fig.~\ref{FW9}.

That each subgraph is induced by a walk allows to automate the production of the formulas, as for instance seen in~\cite{alon1997cycles,perepeshko2009cycle} to count cycles with the adjacency matrix. The correction terms in the following formulas---of the form $-f_F(B)$---are obtained automatically by counting the number of walks of the same length inducing $F$ (which we implemented very inefficiently by enumerating all possible closed walks of length $k$ in $F$ and keeping only those visiting all edges and vertices). 

We write $\diag(A)$ for the column vector containing the main diagonal of $A$ and $\Sigma A$ for the sum of the elements of $A$. To avoid too many repetitions in the formula below, and confusion with the previous set of formulas, we introduce the following matrices: $B_{tr}=B\cdot B_2^\top$, $B_{sq}= B^2B^{2\top}$, $B_{pn} = B^2\cdot B^{3\top}$, $B_{trtr} = \diag(B^3)\diag(B^3)^\top$, $B_{sqtr} = \diag(B^3)\diag(B^4)^\top$, $B_{sql} = B_{sq}\cdot \underline B^2$, $B_{sqr} = B_{sq}\cdot \overline B^2$, $B_{sqtt} = \overline B_{sq}\cdot B^3$, $B_{sqd} = (B^4_{e\bar e})_{e\in E(G)}{\bf 1}^\top$, $B_{pnt} = B_{pn}\cdot \underline B^2\cdot  \overline B^2$.

\begin{align*}
        f_{F1} : B&\mapsto {\bf 1}^\top \diag(B^3)/6\\
        f_{F2} : B&\mapsto {\bf 1}^\top\binom{\diag(B^3)}{2}/2\\
        f_{F3} : B&\mapsto [\Sigma(B\cdot B_{trtr})\\
                                &\qquad\qquad-12f_{F2}(B)-6f_{F1}(B)]/8\\
        f_{F4} : B&\mapsto [\Sigma(\diag(B^3)\cdot \diag(B^4))\\
                                &\qquad\qquad-8f_{F2}(B)]/2\\
        f_{F5} : B&\mapsto {\bf 1}^\top\binom{\diag(B^3)}{3}/2\\
        f_{F6} : B&\mapsto [\Sigma(B\cdot B_{sqtr})\\
                                &\qquad\qquad-6f_{F4}(B)-24f_{F5}(B)-16f_{F2}(B)]/4\\
        f_{F7} : B&\mapsto {\bf 1}^\top B^\boxtimes{\bf 1}/6\\
        f_{F8} : B&\mapsto \Sigma B_{sq}/8\\
        f_{F9} : B&\mapsto \Sigma(B_{sqr}\cdot B^3)/2\\
        f_{F10} : B&\mapsto \Sigma B_{sqtt}/12\\
        f_{F11} : B&\mapsto \Big[{\bf 1}^\top\binom{\diag(B^4)}{2}\\
                                &\qquad\qquad-2f_{F9}(B)-12f_{F10}(B)-12f_{F7}(B)\Big]/2\\
        f_{F12} : B&\mapsto \Sigma\left(B_{sq}\cdot \binom{\overline B^3}{2}\right)/24\\
        f_{F13} : B&\mapsto \Sigma\left(B_{sq}\cdot \binom{B^3}{2}\right)/2\\
        f_{F14} : B&\mapsto [\Sigma(\diag(B^3)\cdot \diag(B^5))\\
                                &\qquad\qquad-4f_{F4}(B)-2f_{F9}(B)]/2\\
        f_{F15} : B&\mapsto [\Sigma(B\cdot B_{trtr})\\
                                &\qquad\qquad-8f_{F8}(B)-48f_{F10}(B)-48f_{F12}(B)\\
                                &\qquad\qquad-12f_{F11}(B)-16f_{F13}(B)-20f_{F9}(B)\\
                                &\qquad\qquad-72f_{F7}(B)]/8\\
        f_{F16} : B&\mapsto [\Sigma((1-B^\top)\cdot B^2\cdot B_{trtr})\\
                                &\qquad\qquad-12f_{F2}(B)-24f_{F7}(B)-16f_{F3}(B)\\
                                &\qquad\qquad-6f_{F9}(B)]/8\\
        f_{F17} : B&\mapsto [\Sigma(B\cdot(\diag(B^3)\diag(B^5)^\top))\\
                                &\qquad\qquad-6f_{F14}(B)-8f_{F4}(B)-10f_{F9}(B)\\
                                &\qquad\qquad-16f_{F13}(B)]/4\\
        f_{F18} : B&\mapsto \Sigma B_{pnt}/4\\
        f_{F19} : B&\mapsto \Sigma(B_{pnt}\cdot B^{2\top})/4\\
        f_{F20} : B&\mapsto \Sigma\left(B^2\cdot \binom{B^{3\top}}{2}\right)/2\\
        f_{F21} : B&\mapsto [\Sigma(B_{sqr}\cdot B_{sqd})\\
                                &\qquad\qquad-8f_{F2}(B)-4f_{F9}(B)-24f_{F5}(B)]/4\\
        f_{F22} : B&\mapsto [\Sigma(B_{sql}\cdot B_{sqd})\\
                                &\qquad\qquad-48f_{F7}(B)-4f_{F9}(B)-8f_{F19}(B)]/4\\
        f_{F23} : B&\mapsto \Sigma(b_{sqr}\cdot B^4)\\
                                &\qquad\qquad-24f_{F7}(B)-2f_{F9}(B)-4f_{F19}(B)\\
        f_{F24} : B&\mapsto [\Sigma(b_{sqtt}\cdot B^3)\\
                                &\qquad\qquad-4f_{F18}(B)]/2\\
        f_{F25} : B&\mapsto [\Sigma(b_{sql}\cdot B^3\cdot (\diag(B^3){\bf 1}^\top))\\
                                &\qquad\qquad-2f_{F9}(B)-8f_{F19}(B)]/6\\
        f_{F26} : B&\mapsto [\Sigma(\diag(B^3)\cdot \diag(B^6))\\
                                &\qquad\qquad-6f_{F1}(B)-20f_{F2}(B)-24f_{F3}(B)-12f_{F5}(B)\\
                                &\qquad\qquad-48f_{F7}(B)-12f_{F9}(B)-4f_{F14}(B)-8f_{F19}(B)\\
                                &\qquad\qquad-8f_{F21}(B)-4f_{F22}(B)-2f_{F23}(B)]/2\\
        f_{F27} : B&\mapsto [\Sigma(\diag(B^4)\cdot \diag(B^5))\\
                                &\qquad\qquad-6f_{F4}(B)-8f_{F20}(B)-8f_{F18}(B)-4f_{F19}(B)\\
                                &\qquad\qquad-2f_{F23}(B)-4f_{F24}(B)]/2\\
        f_{F28} : B&\mapsto [\Sigma((1-B^\top)\cdot B^2\cdot B_{sqtr})\\
                                &\qquad\qquad-12f_{F2}(B)-8f_{F4}(B)-36f_{F5}(B)-8f_{F6}(B)\\
                                &\qquad\qquad-24f_{F7}(B)-6f_{F9}(B)-12f_{F18}(B)-4f_{F19}(B)\\
                                &\qquad\qquad-4f_{F21}(B)-3f_{F23}(B)-4f_{F24}(B)]/4\\
        f_{F29} : B&\mapsto [\Sigma(B\cdot B^4\cdot \underline B^3\cdot \overline B^3)\\
                                &\qquad\qquad-8f_{F2}(B)-16f_{F3}(B)-24f_{F5}(B)-24f_{F7}(B)\\
                                &\qquad\qquad-14f_{F9}(B)-28f_{F21}(B)]/48\\
        f_{F30} : B&\mapsto \Big[\Sigma\left(\binom{\diag(B^3)}{2}\cdot \diag (B^5)\right)\\
                                &-4f_{F18}(B)-8f_{F19}(B)-2f_{F23}(B)-6f_{F25}(B)\Big]/2\\
        f_{F31} : B&\mapsto \Big[\Sigma\left(\binom{\diag(B^4}{2}\cdot \diag(B^3))\right)\\
                                &\qquad\qquad-12f_{F5}(B)-24f_{F7}(B)-2f_{F9}(B)\\
                                &\qquad\qquad-8f_{F18}(B)-10f_{F19}(B)-2f_{F23}(B)\\
                                &\qquad\qquad-2f_{F24}(B)-6f_{F25}(B)\Big]/2\\
        f_{F32} : B&\mapsto [\Sigma(B^2\cdot B^3\cdot B_{trtr})\\
                                &\qquad\qquad-4f_{F2}(B)-12f_{F5}(B)-72f_{F7}(B)\\
                                &\qquad\qquad-14f_{F9}(B)-24f_{F19}(B)-8f_{F21}(B)\\
                                &\qquad\qquad-8f_{F22}(B)-24f_{F25}]/8\\
        f_{F33} : B&\mapsto \Sigma\left(B^{2\top}\cdot \binom{B^3}{3}\right)/2\\
        f_{F34} : B&\mapsto \Big[\Sigma(B\cdot(\diag(B^3)\diag(B^6)^\top))\\
                                &\qquad\qquad-6f_{F1}(B)-36f_{F2}(B)-48f_{F3}(B)-36f_{F5}(B)\\
                                &\qquad\qquad-144f_{F7}(B)-44f_{F9}(B)-8f_{F14}(B)-48f_{F19}(B)\\
                                &\qquad\qquad-48f_{F21}(B)-16f_{F22}(B)-10f_{F23}(B)-6f_{F26}(B)\\
                                &\qquad\qquad-48f_{F29}(B)-16f_{F30}(B)-8f_{F31}(B)-16f_{F32}(B)\\
                                &\qquad\qquad-60f_{F25}(B)\Big]/4\\
        f_{F35} : B&\mapsto \Big[\Sigma(B\cdot (\diag(B^4)\diag(B^5)^\top))\\
                                &\qquad\qquad-10f_{F4}(B)-16f_{F20}(B)-24f_{F18}(B)-28f_{F19}(B)\\
                                &\qquad\qquad-10f_{F23}(B)-20f_{F24}(B)-6f_{F27}(B)-8f_{F30}(B)\\
                                &\qquad\qquad-16f_{F31}(B)-24f_{F33}(B)-36f_{F25}(B)\Big]/4\\
        f_{C_3} : B&\mapsto f_{F1}(B)\\
        f_{C_4} : B&\mapsto \Sigma(\diag(B^4))/8\\
        f_{C_5} : B&\mapsto \Sigma(\diag(B^5))/10\\
        f_{C_6} : B&\mapsto [\Sigma(\diag(B^6))\\
                                &\qquad\qquad-6f_{F1}(B)-12f_{F2}(B)-24f_{F3}(B)]/12\\
        f_{C_7} : B&\mapsto [\Sigma(B^6\cdot B^\top)\\
                                &\qquad\qquad-28f_{F2}(B)-14f_{F4}(B)-84f_{F5}(B)-28f_{F6}(B)]/14\\
        f_{C_8} : B&\mapsto [\Sigma(B^6\cdot B^{2\top})\\
                                &\qquad\qquad-144f_{F7}(B)-8f_{F8}(B)-16f_{F4}(B)-64f_{F9}(B)\\
                                &\qquad\qquad-48f_{F10}(B)-16f_{F11}(B)-96f_{F12}(B)-96f_{F13}(B)\\
                                &\qquad\qquad-16f_{F14}(B)-32f_{F15}(B)-32f_{F16}(B)-32f_{F17}(B)]/16\\
        f_{C_9} : B&\mapsto [\Sigma(B^6\cdot B^{3\top})\\
                                &\qquad\qquad-6f_{F1}(B)-36f_{F2}(B)-72f_{F3}(B)-18f_{F4}(B)\\
                                &\qquad\qquad-36f_{F5}(B)-288f_{F7}(B)-90f_{F9}(B)-18f_{F14}(B)\\
                                &\qquad\qquad-108f_{F18}(B)-180f_{F19}(B)-36f_{F20}(B)\\
                                &\qquad\qquad-108f_{F21}(B)-36f_{F22}(B)-72f_{F23}(B)\\
                                &\qquad\qquad-72f_{F24}(B)-288f_{F25}(B)-18f_{F26}(B)\\
                                &\qquad\qquad-18f_{F27}(B)-36f_{F28}(B)-144f_{F29}(B)\\
                                &\qquad\qquad-108f_{F30}(B)-108f_{F31}(B)-72f_{F32}(B)\\
                                &\qquad\qquad-108f_{F33}(B)-36f_{F34}(B)-36f_{F35}(B)]/18
\end{align*}
\bibliographystyle{unsrt}%
\bibliography{Profiler_bib}
\end{document}